\documentclass[pdflatex,sn-mathphys-num]{sn-jnl}

\usepackage{graphicx}%
\usepackage{multirow}%
\usepackage{amsmath,amssymb,amsfonts}%
\usepackage{amsthm}%
\usepackage{mathrsfs}%
\usepackage[title]{appendix}%
\usepackage{xcolor}%
\usepackage{textcomp}%
\usepackage{manyfoot}%
\usepackage{booktabs}%
\usepackage{algorithm}%
\usepackage{algorithmicx}%
\usepackage{algpseudocode}%
\usepackage{listings}%

\usepackage{tabularx}
\usepackage{array}
\usepackage{booktabs}
\usepackage{subcaption}

\newcolumntype{Y}{>{\centering\arraybackslash}X}

\theoremstyle{thmstyleone}%
\newtheorem{theorem}{Theorem}
\theoremstyle{thmstyletwo}%

\theoremstyle{thmstylethree}%

\begin{document}

\title[Article Title]{Fast and Scalable Influence Estimation in the
Independent Cascade Model: a unified framework}


\author[1]{\fnm{Matic} \sur{Požar}}\email{matic.pozar@upr.si}

\author*[2,3]{\fnm{Miklós} \sur{Krész}}\email{miklos.kresz@innorenew.eu, miklos.kresz@szte.hu}

\affil[1]{\orgdiv{Department of Information Sciences and Technologies}, \orgname{UP FAMNIT}, \orgaddress{\street{Glagoljaška 8}, \city{Koper}, \postcode{6000}, \country{Slovenia}}}

\affil[2]{\orgdiv{InnoRenew CoE}, \orgname{UP IAM and UP FAMNIT, University of Primorska}, \orgaddress{\street{Titov trg 4}, \city{Koper}, \postcode{6000},\country{Slovenia}}}

\affil[3]{\orgdiv{Department of Applied Informatics}, \orgname{University of Szeged}, \orgaddress{\street{Boldogasszony sgt. 6}, \city{Szeged}, \postcode{6725}, \country{Hungary}}}


\abstract{
\begin{abstract}

Computing diffusion processes on networks is fundamental to applications
including influence maximization, epidemic modeling, information
propagation, and diffusion inference. The Independent Cascade (IC) model is
one of the most widely used models of progressive diffusion, yet exact
computation of its activation probabilities is \#P-hard, while Monte Carlo
simulation becomes costly when accurate diffusion estimates are required
repeatedly. Despite its central role in these applications, approximation
of the IC diffusion function itself, particularly through analytical
methods, has received comparatively little systematic attention relative
to the extensive literature on influence-maximization algorithms.

We present a comprehensive study of analytical approximations for the IC
model. We review a wide range of existing approaches and unify them through
a common framework based on temporal representation, retained propagation
information, and aggregation rule. This perspective exposes connections
between previously distinct approximation families and motivates several
new methods.

We further conduct, to our knowledge, the first large-scale cross-family
comparison of analytical IC approximations for node-level marginal
activation probabilities. Our benchmark spans real-world and synthetic
networks with systematically varied topology, seed-set size,
edge-probability model, and diffusion strength. The proposed second-order
dynamic message-passing method achieves the lowest pooled median error and
the lowest dataset- or graph-family-level median error on every network
considered. The results demonstrate substantial and systematic differences
between approximation families, identify the regimes in which higher-order
corrections are most beneficial, and establish analytical IC approximation
as a rich accuracy--efficiency design space rather than a collection of
isolated methods.
\end{abstract}
}




\maketitle

\section*{Abbreviations}

\begin{table}[ht]
\centering
\caption{List of abbreviations.}
\label{tab:abbreviations}
\small
\begin{tabularx}{\textwidth}{@{}l X l X@{}}
\toprule
\textbf{Abbr.} & \textbf{Meaning} & \textbf{Abbr.} & \textbf{Meaning} \\
\midrule
IC   & Independent Cascade
     & GIC  & Generalized Independent Cascade \\

MC   & Monte Carlo
     & IM   & Influence Maximization \\

DMP  & Dynamic Message Passing
     & SWE  & Step-Wise Estimation \\

NO   & Noisy-OR
     & ALE  & Aggregated Linear Effect \\

IPL  & Independent Propagation-Layer
     & NM   & Naive Method \\

SPM  & Shortest Path Model
     & MIA  & Maximum Influence Arborescence \\

SSS  & Steady State Spread
     & NBUB & Nonbacktracking Upper Bound \\

BP   & Belief Propagation
     & IRIE & Influence Ranking and Influence Estimation \\

TDA  & Tree Distribution Approximation
     & RMSE & Root Mean Squared Error \\

MAE  & Mean Absolute Error
     & AUC  & Area Under the Curve \\
\bottomrule
\end{tabularx}
\end{table}

\section{Introduction}
\label{sec:introduction}

Modeling network diffusion dynamics is an important application of network analysis, with a wide range of use cases. Broadly speaking, a diffusion process describes how some quantity of interest propagates between nodes of a network. In this work, we focus on the Independent Cascade (IC) model \cite{GOLDENBERG_IC}, due to its simplicity, flexibility, and widespread adoption. Its canonical application is the Influence Maximization (IM) problem \cite{kempe_im}, which asks for a small set of initially active nodes that maximizes the expected influence spread. However, IC and related network diffusion models arise in a much broader range of applications, including epidemic modeling \cite{KARRER_EPIDEMIC,ALTARELLI_EPIDEMIC,LOKHOV_EPIDEMIC_SOURCE}, diffusion-source inference \cite{ZHAI_SOURCE,RUMOR_SOURCE,BERENBRINK_IC_SOURCE}, opinion dynamics \cite{QUATTROCIOCCHI_OPINION}, diffusion-parameter learning and network inference \cite{SAITO_IC_LEARNING,GOYAL_INFLUENCE_LEARNING,SAITO_MODEL_SELECTION,NETINF}, recommender systems \cite{WU_SOCIAL_RECOMMENDATION}, banking \cite{BOTA_BANK_TRANSACTION}, and modeling the progression of neurodegenerative disease through brain networks \cite{RAJ_NETWORK_DIFFUSION}.

Despite its widespread use, exact evaluation of the IC model is computationally intractable; a reduction from \(st\)-connectedness shows that computing its influence function is \#P-hard \cite{mia_icphard}. Although exact computation methods exist \cite{bdd_influence}, they are
generally computationally practical only on restricted graph classes; for
example, bounded-pathwidth graphs admit linear-time exact computation
\cite{bounded_pathwidth}. In general, a standard approach is Monte Carlo (MC) simulation, in which the diffusion process is repeatedly sampled and activation probabilities or expected influence are estimated empirically. Experimental studies commonly use on the order of \(10^4\) simulations per seed set to obtain stable estimates \cite{kempe_im,irie,dmp,sss}, although substantially larger sample sizes may be required when high precision is needed \cite{burkholz2021cascade,nbub}. More sophisticated sampling schemes can substantially reduce the cost of
estimating the total influence of a fixed seed set \cite{infest,siea}.
However, such methods typically target the scalar expected spread rather
than the complete vector of node activation probabilities. This aggregation
can obscure important structural information: for example, two diffusion
instances may have the same expected spread even if one assigns moderate
activation probability to many nodes, while the other concentrates most of
the probability on a small number of highly exposed nodes. Such differences
are important whenever the identities and relative risks of individual nodes,
rather than only the total number expected to activate, are of interest. Moreover, scalability for a single influence query does not directly translate to repeated high-precision marginal-gain evaluations, as required by greedy IM algorithms, particularly when the marginal gain is small relative to the influence already accumulated by the current seed set.

For the sole purpose of solving the canonical IM problem, reverse-reachable-set and sketch-based methods are among the strongest scalable approaches \cite{skim,tim}. Their scalability arises partly from avoiding explicit reconstruction of the diffusion process: rather than estimating all node-level activation probabilities, they use sampled reachability structures to approximate the influence objective and optimize the seed set directly. Nevertheless, there is empirical evidence that methods modeling the underlying diffusion dynamics more explicitly can, in some settings, obtain seed sets with greater Monte-Carlo-evaluated influence \cite{im_wiht_gnn}. These approaches therefore do not eliminate the need for efficient diffusion estimation, particularly in applications where node-level activation probabilities, repeated evaluation under changing diffusion parameters, or differentiability of the estimator are themselves important.

For these reasons, we believe there is considerable value in studying analytical approximations of the IC process. Learning-based diffusion estimators have also shown promising results
\cite{im_wiht_gnn,nedmp,neural_mean_field}, but we do not include them
in the main comparison because they operate under a substantially different
computational setting in which performance depends not only on the specified
diffusion instance, but also on the availability and representativeness of
training data, the learning architecture, and generalization to unseen graph
and diffusion regimes. A direct comparison with sampling-free analytical methods would therefore require a separate experimental protocol. Determining when learned corrections provide sufficient accuracy gains to justify these additional requirements remains an interesting direction for future work.

Despite the large number of analytical IC approximations proposed in the literature, to the best of our knowledge no previous work has compared the major approaches head-to-head across a broad and controlled experimental environment. Existing works frequently evaluate approximations only as components of downstream tasks such as IM, making it difficult to isolate the quality of the diffusion estimator itself. Moreover, standard IM evaluations often employ relatively low edge probabilities, where diffusion is weak and differences between approximation methods can be less pronounced. A controlled comparison across graph topologies and diffusion regimes can therefore reveal behavior that is difficult to observe from downstream IM performance alone.

Our main contributions are:
\begin{itemize}
\item We provide a systematic taxonomy of analytical IC approximation methods in terms of the dependencies and feedback effects they neglect or explicitly model.
\item We develop a unified mathematical formulation that exposes relationships between several existing methods and motivates new approximations.
\item We propose several new analytical approximation methods and extensions, including higher-order cavity approximations. In our held-out evaluation, second-order DMP achieves the lowest estimation error on every tested dataset.
\item We conduct an extensive empirical comparison on synthetic and real-world networks across a broad range of graph structures, seed-set sizes, edge-probability models, and diffusion strengths.
\end{itemize}

To our knowledge, this work provides the first comprehensive study of analytical approximation methods for the Independent Cascade model. Rather than evaluating a single approximation or a narrowly related family, we systematically organize and compare methods spanning cumulative, step-wise, path-based, linear-propagation, cavity/message-passing, and higher-order formulations. The methods are evaluated under a common experimental protocol for node-level marginal activation probabilities across diverse real-world and synthetic networks, diffusion strengths, seed-set sizes, and edge-probability models.

The remainder of this paper is organized as follows. Section~\ref{sec:problem-formulation} introduces the required formal preliminaries. Section~\ref{sec:unified-view} describes representative approximation methods, discusses the principal sources of approximation error, and develops a unified mathematical view. Section~\ref{sec:new-methods} presents our new approximation methods, motivated by the unified formalism and the limitations identified in the preceding section. Section~\ref{sec:experimental-methodology} describes the experimental protocol and implementation details, and Section~\ref{sec:results} presents the resulting empirical analysis. Section~\ref{sec:discussion} discusses the broader implications and limitations of the methods, before Section~\ref{sec:conclusion} concludes the paper.

A preliminary version of aspects of this work was presented at the 6th International Symposium on Complex Systems (ISCS 2026)~\cite{pozar2026approximation}. The present work substantially extends the conference version with additional approximation methods, theoretical analysis, and a broader experimental evaluation.


\section{Problem Formulation}
\label{sec:problem-formulation}

In this section we introduce the formalism required to describe diffusion under the Independent Cascade (IC) model. We then define the node-level marginal activation probabilities and expected influence spread that form the main quantities of interest in this work. Finally, we discuss Monte Carlo estimation and briefly introduce the Influence Maximization (IM) problem, which provides an important application of efficient IC estimation.

\subsection{Networks and the Independent Cascade Model}
\label{subsec:ic-model}

We use standard graph-theoretic notation. Let \(G=(V,E)\) be a simple directed graph, where \(V\) is the set of nodes and \(E\subseteq V\times V\) is the set of directed edges. We use the terms \emph{node} and \emph{vertex}, as well as \emph{graph} and \emph{network}, interchangeably. Undirected networks are represented as bidirectional graphs, where an undirected edge \(\{u,v\}\) is represented by the two directed edges \((u,v)\) and \((v,u)\).

For a node \(u\in V\), we denote its in- and out-neighborhoods by

$$
N^{\mathrm{in}}(u)
=
\{v\in V:(v,u)\in E\},
\qquad
N^{\mathrm{out}}(u)
=
\{v\in V:(u,v)\in E\}.
$$

When the distinction is unnecessary, we write \(N(u)\). Analogously, for a set \(S\subseteq V\), we define

$$
N^{\mathrm{out}}(S)
=
\bigcup_{u\in S}N^{\mathrm{out}}(u)\setminus S.
$$

When it is useful to make the underlying graph explicit, we use an additional subscript, e.g., \(N_G^{\mathrm{out}}(u)\).

Under the IC model, every node is in one of two states: active or inactive. Active nodes remain active permanently. Each directed edge \((u,v)\in E\) is associated with an activation probability

$$
p_{uv}\in[0,1],
$$

representing the probability with which a newly activated node \(u\) successfully activates an inactive out-neighbor \(v\). We collect these probabilities in the weighted propagation matrix
\(B\in[0,1]^{|V|\times |V|}\), oriented according to the column-vector
convention, where
\[
B_{vu}
=
\begin{cases}
p_{uv}, & (u,v)\in E,\\
0, & (u,v)\notin E.
\end{cases}
\]
Consequently,
\[
(Bx)_v=\sum_{u:(u,v)\in E}p_{uv}x_u.
\]

The diffusion begins from an initially active set \(S\subseteq V\), called the \emph{seed set}. Let \(A(t)\) denote the set of nodes that become newly active at time \(t\), with

$$
A(0)=S.
$$

At every subsequent time step, each node \(u\in A(t)\) independently attempts to activate each currently inactive out-neighbor \(v\), succeeding with probability \(p_{uv}\). Each edge is therefore given at most one activation attempt during a realization of the process. The process terminates when no new node becomes active and, since every non-terminal step activates at least one previously inactive node, termination occurs after at most \(|V|-1\) propagation steps.

An equivalent and particularly useful formulation of the IC model is given by the \emph{live-edge representation}. Before the diffusion begins, we independently sample every edge \(e=(u,v)\in E\) according to

$$
X_e\sim\operatorname{Bernoulli}(p_e).
$$

An edge is called \emph{live} when \(X_e=1\). One particular realization of the IC process therefore induces a random unweighted subgraph

$$
G'=(V,E'),
\qquad
E'=\{e\in E:X_e=1\}.
$$

Given \(G'\), the diffusion itself becomes deterministic. In particular,

$$
A(t+1)
=
N^{\mathrm{out}}_{G'}(A(t))
\setminus
\bigcup_{\ell\leq t}A(\ell).
$$

Equivalently,

$$
A(t)
=
\{u\in V:d_{G'}(S,u)=t\},
$$

where

$$
d_{G'}(S,u)
=
\min_{s\in S}d_{G'}(s,u),
$$

and \(d_{G'}(s,u)\) denotes the length of the shortest directed path from \(s\) to \(u\) in \(G'\), with \(d_{G'}(s,u)=\infty\) when no such path exists. Thus, a node is eventually activated if and only if it is reachable from the seed set in the sampled live-edge graph.

The standard IC model assumes a deterministic seed set. More generally, initial activation itself can be stochastic. Let

$$
\boldsymbol{\pi}^{(0)}
=
\{\pi_u^{(0)}\}_{u\in V}
$$

denote initial activation probabilities, and include every node \(u\) in the seed set independently according to

$$
u\in S
\sim
\operatorname{Bernoulli}\!\left(\pi_u^{(0)}\right).
$$

This extension is commonly referred to as the Generalized Independent Cascade (GIC) model \cite{generalized_ic} and is useful in settings where an intervention does not activate a selected node deterministically. For ease of presentation, we primarily use deterministic seed sets throughout the paper; the extension of the considered approximation methods to probabilistic initial activations is generally straightforward.

\subsection{Marginal Activation Probabilities and Influence Spread}
\label{subsec:marginals}

The IC model induces a probability distribution over live-edge graphs and, consequently, over possible activation cascades. We are generally interested not in a particular realization, but in the expected outcome.

Let \(p_u(t)\) denote the marginal probability that node \(u\) has become active by time \(t\):

$$
p_u(t)
=
\Pr\left[
u\in\bigcup_{\ell\leq t}A(\ell)
\right].
$$

Using the live-edge formulation, this is equivalently

$$
p_u(t)
=
\Pr_{G'}
\left[
d_{G'}(S,u)\leq t
\right],
$$

or, in expectation form,

$$
p_u(t)
=
\mathbb{E}_{G'}
\left[
\mathbf{1}
\left\{
d_{G'}(S,u)\leq t
\right\}
\right],
$$

where the expectation is taken over independently sampled live-edge realizations \(G'\).

For some of the approximation methods considered later, it is useful to introduce \emph{cavity marginals}. For a set

$$
R\subseteq V\setminus(S\cup\{u\}),
$$

we define

$$
p_{u\rightarrow R}(t)
=
\Pr_{G'}
\left[
d_{G'\setminus R}(S,u)\leq t
\right].
$$

Thus, \(p_{u\rightarrow R}(t)\) is the activation probability of \(u\) in the \emph{cavity graph} obtained by removing the nodes in \(R\). We use the shorthand

$$
p_{u\rightarrow v}(t)
=
p_{u\rightarrow\{v\}}(t),
$$

and analogously for cavity sets of small cardinality. Importantly, this quantity is defined through node removal rather than by conditioning on the event that the nodes in \(R\) happen to remain inactive in the original stochastic process. 

Some models additionally track probabilities of nodes becoming activated at particular steps, rather than cumulating the overall influence over several steps. We define
\[
r_u(t) = \Pr
\left[
d_{G'}(S,u)= t
\right],
\]
i.e. the probability that the shortest length path from any node in the seed set $S$ to node \(u\) is exactly \(t\) in the live edge graph $G'$.

We denote the final activation probability of node \(u\) by

$$
p_u^\infty
=
\Pr_{G'}
\left[
d_{G'}(S,u)<\infty
\right].
$$

Since the IC process terminates after finitely many steps,

$$
p_u^\infty
=
p_u(|V|-1).
$$

A commonly used aggregate quantity is the expected influence spread of a seed set \(S\), defined as the expected number of nodes eventually activated:

$$
\sigma(S)
=
\mathbb{E}
\left[
\left|
\bigcup_{t\geq0}A(t)
\right|
\right].
$$

By linearity of expectation,

$$
\boxed{
\sigma(S)
=
\sum_{u\in V}p_u^\infty.
}
$$

The influence function therefore aggregates the full vector of node-level activation probabilities

$$
\mathbf{p}^{\infty}
=
\left(p_u^\infty\right)_{u\in V}
$$

into a single scalar. While this quantity is sufficient for some downstream tasks, most prominently standard IM, it discards the fine-grained information contained in the individual activation probabilities.

\subsection{Monte Carlo Estimation and Computational Complexity}
\label{subsec:mc}

Exact computation of IC marginals is computationally difficult because exponentially many live-edge realizations can contribute to the activation probability of a node. Indeed, exact IC evaluation is \#P-hard \cite{mia_icphard}. A direct reduction follows from the \(s\)-\(t\) connectedness problem studied by Valiant \cite{VALIANT_RELIABILITY}, which asks for the number of subgraphs of a graph containing an \(s\)-\(t\) path.

Consider an instance of \(s\)-\(t\) connectedness and construct an IC instance on the same graph by setting \(S=\{s\}\) and assigning \(p_e=0.5\) for every \(e\in E\). The graph in question can be considered to be directed or undirected.

Every subgraph of \(G\) is then sampled as a live-edge realization with equal probability \(2^{-|E|}\). The final activation probability of \(t\) is therefore

$$
p_t^\infty
=
\frac{
\#\{E'\subseteq E:(V,E')\text{ contains an }s\text{-}t\text{ path}\}
}{
2^{|E|}
}.
$$

Consequently, an exact computation of \(p_t^\infty\) would yield the solution to the \(s\)-\(t\) connectedness counting problem.

The standard asymptotically exact approach to estimating IC marginals on general networks is Monte Carlo (MC) simulation. Let

$$
\mathcal{G}
=
\{G'_1,\ldots,G'_M\}
$$

be \(M\) independently sampled live-edge realizations. The marginal probability of node \(u\) is estimated as

$$
\hat p_u(t)
=
\frac{1}{M}
\sum_{i=1}^{M}
\mathbf{1}
\left\{
d_{G'_i}(S,u)\leq t
\right\}.
$$

Since every term is a Bernoulli random variable with parameter \(p_u(t)\),

$$
\mathbb{E}[\hat p_u(t)]
=
p_u(t),
$$

and

$$
\operatorname{Var}[\hat p_u(t)]
=
\frac{
p_u(t)(1-p_u(t))
}{M}
\leq
\frac{1}{4M}.
$$

Hence,

$$
\operatorname{SE}[\hat p_u(t)]
\leq
\frac{1}{2\sqrt{M}}.
$$

For example, with \(M=10^4\) independent realizations, the worst-case standard error of an individual node marginal is at most \(0.005\). Increasing the precision further requires a quadratic increase in the number of samples with respect to the desired reduction in standard error.

The influence spread can analogously be estimated as

$$
\hat\sigma(S)
=
\frac{1}{M}
\sum_{i=1}^{M}
\left|
\{u\in V:d_{G'_i}(S,u)<\infty\}
\right|.
$$

Sampling budgets on the order of \(10^4\) realizations are commonly used in empirical IC studies \cite{kempe_im,irie,dmp,sss}, although substantially larger sample sizes may be required when high precision is needed \cite{burkholz2021cascade,nbub}. In particular, unlike a single node indicator, the cascade size takes values between \(0\) and \(|V|\), and its variance can grow with the size and structure of the network \cite{nbub}.

The computational cost of MC estimation scales linearly with the number of sampled realizations. In the worst case, a realization may require exploring a substantial fraction of the graph, yielding a cost of

$$
O\!\left(
M(|V|+|E|)
\right)
$$

for \(M\) complete simulations. Individual simulations are completely independent and hence highly parallelizable, but every realization still requires a stochastic reachability computation. Consequently, MC estimation can become computationally expensive when high precision is required or when influence must be evaluated repeatedly, as is the case in many IM algorithms.

\subsection{Influence Maximization Context}
\label{subsec:im-context}

The Influence Maximization problem is a classical combinatorial optimization problem introduced in the algorithmic diffusion literature by Kempe et al. \cite{kempe_im}. Given a graph \(G\), an IC model, and a seed budget \(k\), the goal is to select a seed set

$$
S^*
=
\arg\max_{\substack{S\subseteq V\\|S|\leq k}}
\sigma(S).
$$

The problem is NP-hard under the IC model \cite{kempe_im}, as can be shown by a reduction from Set Cover.

Kempe et al. further showed that the IC influence function is monotone and submodular \cite{kempe_im}. Monotonicity means that

$$
S\subseteq T
\quad\Longrightarrow\quad
\sigma(S)\leq\sigma(T),
$$

while submodularity corresponds to diminishing marginal returns:

$$
\sigma(S\cup\{v\})-\sigma(S)
\geq
\sigma(T\cup\{v\})-\sigma(T),
\qquad
S\subseteq T,
\quad
v\notin T.
$$

Consequently, when exact influence values are available, the classical greedy algorithm achieves the approximation guarantee

$$
\sigma(S_{\mathrm{greedy}})
\geq
\left(1-\frac{1}{e}\right)
\sigma(S^*).
$$

The greedy method is given in Algorithm~\ref{alg:greedy-im}.

\begin{algorithm}[H]
\caption{Greedy Influence Maximization}
\label{alg:greedy-im}
\begin{algorithmic}[1]
\Require Graph $G=(V,E)$, seed budget $k \le |V|$, influence function $\sigma$
\Ensure Seed set $S$ with $|S|=k$

\State $S\gets\emptyset$

\For{$i=1,\ldots,k$}
\State $v^\star\gets
\arg\max\limits_{v\in V\setminus S}
\left[
\sigma(S\cup{v})-\sigma(S)
\right]$
\State $S\gets S\cup{v^\star}$
\EndFor

\State \Return $S$
\end{algorithmic}
\end{algorithm}

A direct implementation requires evaluating the influence of every remaining candidate at every iteration, resulting in

$$
\sum_{i=0}^{k-1}(|V|-i)
=
O(k|V|)
$$

candidate influence evaluations. Since exact computation of \(\sigma\) is \#P-hard, practical implementations replace it with an estimator \(\hat\sigma\). Lazy-forward strategies such as CELF \cite{CELF} can substantially reduce the number of repeated evaluations in practice, but influence estimation nevertheless remains a central computational primitive in greedy IM methods. Moreover, when candidate marginal gains are small relative to the already accumulated influence, distinguishing between candidates may require particularly accurate influence estimates.

The goal of this work is therefore to study efficient analytical approximations of the node-level marginal activation probabilities induced by the IC process. These marginals provide a richer description of the diffusion than the aggregate influence alone, while the influence function follows immediately as

$$
\sigma(S)
=
\sum_{u\in V}p_u^\infty.
$$

Our primary objective is to obtain such estimates at a computational cost substantially below that of high-precision Monte Carlo simulation.


\section{A Unified View of Independent Cascade Approximation}
\label{sec:unified-view}

Throughout Sections~\ref{sec:unified-view} and \ref{sec:new-methods}, $N(u)$ denotes the in-neighborhood $N^{\mathrm{in}}(u)$; for an undirected graph it reduces to the ordinary neighborhood.

\subsection{Exact Marginal Formulation and Main Challenges}
\label{subsec:exact-marginal}

The marginal probabilities can be expressed exactly as
\begin{equation} \label{eq:exact_marginals}
p_u(t)=p_u+(1-p_u)\sum_{R\subseteq N(u)}p_{R\rightarrow u}(t-1)\left [1-\prod_{v\in R}(1-p_{vu})\right ],
\end{equation}

where the quantity \(p_{R\rightarrow u}(t-1)\) denotes the probability that \emph{exactly} the neighbors of \(u\) in \(R\) are activated by time \(t-1\) in the cavity graph \(G\setminus \{u\}\) where node \(u\) is removed.
Note that we use $p_u:=p_u(0)=\pi_u^{(0)}=p^{(0)}$ interchangeably for denoting the initial activation probabilities at time $t=0$.
Beyond the fact that all possible neighbor states are needed for the marginal \(p_u(t)\), the state probabilities \(p_{R\rightarrow u}(t-1)\) have no canonical form, since their activations can be correlated in arbitrary ways. In particular, whenever two nodes share at least one node in their sets of predecessors, their states can become dependent.

Secondly, the removal of contribution of node \(u\) in Equation~\ref{eq:exact_marginals} is necessary, since if any one of the neighbors \(N(u)\) activates \(u\), then \(u\) itself must not have contributed to the activation of that neighbor in any way, otherwise we get undesired feedback loops, whereby a node could influence itself. More precisely, any live path witnessing activation of $u$ can be taken to be simple; its prefix before the final edge into $u$ therefore does not use $u$. If this is not considered the computed probabilities would artificially be larger.

We note that the two sources of difficulties are in some sense orthogonal, and each contributes to the general computational difficulties of the marginals. First, consider instances of the IC in directed acyclic graphs. Clearly feedback loops play no role in such instances, furthermore one can remove the cavity messages in Equation~\ref{eq:exact_marginals}. However, the \(s\)-\(t\) connectedness problem is \#P-hard even for directed acyclic planar graphs of maximum vertex degree three \cite{PROVAN_DAG_RELIABILITY}.

Secondly, under the assumption that neighbor activations are independent, Equation~\ref{eq:exact_marginals} reduces to

\[
\hat p_u(t) = 1-(1-p_u)\prod_{v\in N(u)}\left [1-p_{vu} \cdot \hat p_{v\rightarrow u}(t-1) \right ]
\]
while the general cavity messages are given by
\begin{equation}\label{eq:independent_general_cavity}
    \hat p_{u\rightarrow S}(t)=1-(1-p_u)\prod_{v\in N(u)\setminus S}\left [ 1-p_{vu}\cdot \hat p_{v\rightarrow S\cup \{u\}}(t-1) \right ].
\end{equation}
The interpretation of Equation~\ref{eq:independent_general_cavity} is that node \(u\) can get activated by time \(t\) in the graph where nodes in \(S\) are removed, if it is initially activated, or if any one of the neighbors of \(u\) that are not in \(S\) activates \(u\) at a time before \(t\), crucially without the help of any of the nodes in \(S\cup \{u\}\). The states of the cavity messages are indexed by entire subsets of nodes, hence in the worst case there can be \(n2^{n-1}\) such cavity states. With memoization, the most obvious evaluation therefore has exponential state complexity, roughly \(O(\text{poly}(n)2^n)\).


%
%

\subsection{Existing Solutions}
\label{subsec:existing-solutions}

In general there are no computationally inexpensive exact solutions to the two sources of difficulty; state correlations and feedback loops, hence they are often ignored. If both properties are fully relaxed, one obtains the following simplification of Equation~\ref{eq:exact_marginals}
\begin{equation} \label{eq:Naive-Method}
\hat p_u(t) = 1-(1-p_u)\prod_{v\in N(u)}\left [1-p_{vu}\cdot \hat p_v(t-1)\right ].
\end{equation}
We shall refer to Equation~\ref{eq:Naive-Method} as the \textbf{Naive Method} (NM). Many existing IC approximation solutions explicitly rely on the NM. 

For example, SPM and SP1M \cite{SPM} models essentially use the same update rule as the NM and consider only paths of shortest length (or additionally paths of shortest length plus one), namely
\[
\operatorname{SPM}(u) = \operatorname{NM}_u(d(S,u))
\]
and
\[
\operatorname{SP1M}(u) = \operatorname{SPM}(u) + (1-\operatorname{SPM}(u))\cdot \operatorname{NM}_u(d(S,u)+1).
\]
This helps mitigate some of the issues discussed thus far but at the same time explicitly neglects much of the dynamics, which can lead to poor performance especially in larger edge probability settings. 

The neighborhood bound heuristic \cite{ale} uses a very similar update rule but restrict the influence propagation only to a small neighborhood of each node. 

The maximum influence
arborescence \cite{mia_icphard} (MIA) approach uses the same update rule as the NM, however only after restricting to an arborescence, not on the original graph. This is done by retaining only the most influential paths for each node. For each target node \(v\), they find the maximum influence path
\[
\operatorname{MIP}(u,v) = \arg\max_{P\in\mathcal P(u,v)}\prod_{e\in P}p_e
\]
from every possible source \(u\) to \(v\). Here $\mathcal P(u,v)$ is the set of simple directed $u$-to-$v$ paths. They discard paths whose probability is below a threshold \(\theta\), then take the union of the surviving maximum influence paths into \(v\). This gives the maximum influence in-arborescence \(\operatorname{MIIA}(v,\theta)\). The NM is then computed for each node \(v\) on \(\operatorname{MIIA}(v,\theta)\).

Several approaches for approximating the IC are introduced in \cite{sss}, all strongly related to the NM. The Steady State Spread (SSS) iterates Equation~\ref{eq:Naive-Method} until a convergence criterion is met
\[
\operatorname{SSS}(u)=\lim_{t\rightarrow \infty}\operatorname{NM}_u(t)
\]
Due to the issues imposed by feedback loops, the method SSS-Noself is introduced, where cavity-like quantities are introduced by computing influence spread in auxiliary graphs where certain edge probabilities are set to zero, note that this is done for each edge so the overhead is significant.
They define \(\pi_u^{[v]}\) as the probability that node \(u\) gets activated in the cavity graph where node \(v\) is removed. They compute this by the SSS estimation for node \(u\) in the IC instance where all edges incident to \(v\) have probability zero.
Finally, SSS-Bounded-Path is introduced where the number of time steps of NM is limited to some finite value instead of iterating until a steady state is reached. 

The nonbacktracking upper bound (NBUB) from \cite{nbub} can be interpreted as replacing the node-level propagation underlying NM by edge-specific nonbacktracking messages. 

Dynamic Message Passing \cite{dmp,dmp_unidirectional} (DMP) directly addresses the difficulty of feedback loops, while it explicitly keeps the neighbor activations independent in their modeling. The key idea is quite simple, namely to apply the well known Belief Propagation \cite{belief_propagation} (BP) framework to the setting of the IC. DMP applies the dynamic-cavity/BP principle to IC; algebraically, it replaces the node-level noisy-OR propagation of NM by edge-cavity noisy-OR messages. They estimate the cavity quantities for each edge as
\[
\hat p_{u\rightarrow v}(t)=1-(1-p_u)\prod_{w\in N(u)\setminus \{v\}}[1-p_{wu}\cdot \hat p_{w\rightarrow u}(t-1)].
\]
The crucial modeling assumption made is 
\[
p_{w\rightarrow u,v}(t-1)\approx p_{w\rightarrow u}(t-1).
\]
Then the final marginals are given as 
\[
\hat p_u(t)=1-(1-p_u)\prod_{v\in N(u)}[1-p_{vu}\cdot \hat p_{v\rightarrow u}(t-1)].
\]
DMP is exact when the underlying undirected graph is a tree (or forest), with independent edge variables and independent/deterministic seeds. Somewhat surprisingly, DMP is rarely used in the IM literature, as we shall see in upcoming sections, this in empirically not warranted. It is worth noting that the belief-propagation marginals used by \cite{burkholz2019,burkholz2021cascade} reduce, for the IC model, to the same cavity fixed-point equations subsequently used by \cite{dmp} for infinite-time IC influence estimation. Additionally, \cite{irie} uses the same kind of belief propagation ideas, however for the purpose of estimating the influence spread function rather than directly estimating marginals.

While less common, some methods explicitly approximate the probability that
a node becomes active at a particular time step and subsequently aggregate
these quantities. Let \(r_u(t)\) denote the probability that node \(u\)
becomes newly active at time \(t\). Two events must then co-occur:
\(u\) must remain inactive through time \(t-1\), and at least one neighbor
that became newly active at time \(t-1\) must successfully activate \(u\)
at time \(t\).

Let \(S_u(t-1)\) denote the event that \(u\) is inactive through time
\(t-1\), and let \(F_u(R,t-1)\) denote the event that exactly the nodes
in \(R\subseteq N(u)\) became newly active at time \(t-1\). Then
\begin{equation}
\label{eq:exact_swe_marginals}
r_u(t)
=
(1-p_u(t-1))
\sum_{R\subseteq N(u)}
\Pr\!\left[
F_u(R,t-1)\mid S_u(t-1)
\right]
\left[
1-\prod_{v\in R}(1-p_{vu})
\right].
\end{equation}

Unlike the cavity formulation of Equation~\ref{eq:exact_marginals}, the
dependency on \(u\) is represented here through conditioning on its
survival rather than by removing \(u\) from the graph. The computational
difficulty is therefore contained in the conditional joint distribution
of newly active neighbors as well as in the incorporation of conditional information. In contrast to cavity information, the conditioning in Equation~\ref{eq:exact_swe_marginals} does not admit an analogous canonical local recursion. A cavity condition can be represented structurally by removing a set of nodes, leading naturally to the hierarchy of messages \(p_{u\rightarrow S}\) in Equation~\ref{eq:independent_general_cavity}. The event \(S_u(t-1)\), however, is an observation that \(u\) has remained inactive despite all possible activation mechanisms up to time \(t-1\). This event can arise from many different combinations of inactive predecessors and failed transmissions, and conditioning on it consequently reweights the activation histories of other nodes in a topology- and time-dependent manner. Exact incorporation of this information therefore generally requires joint state or trajectory probabilities rather than a simple enlargement of a cavity set.

The canonical step-wise estimator is obtained by replacing this
conditional joint distribution with a product of unconditional
step-wise marginals,
\[
\Pr[F_u(R,t-1)\mid S_u(t-1)]
\approx
\prod_{v\in R}\hat r_v(t-1)
\prod_{v\in N(u)\setminus R}
(1-\hat r_v(t-1)).
\]
Substituting this approximation into
Equation~\ref{eq:exact_swe_marginals} gives
\[
\hat r_u(t)
=
(1-\hat p_u(t-1))
\left[
1-
\prod_{v\in N(u)}
(1-p_{vu}\hat r_v(t-1))
\right],
\]
with cumulative marginals
\begin{equation}
\label{eq:swe}
\hat p_u(t)=\sum_{i\leq t}\hat r_u(i).
\end{equation}
We refer to this approximation as \textbf{Step-Wise Estimation} (SWE).

An equivalent update arises from the unified influence framework of
\cite{unified_influence_model} when specialized to the Independent
Cascade model. Closely related equations were derived by
\cite{zhang_probabilistic_2013} using an inclusion--exclusion formulation.

A particularly useful interpretation is provided by
\cite{chen_taylor_2017}. They formulate the exact IC transition in terms
of the conditional neighborhood distribution
\[
P(X_{N(u)}(t)=x\mid X_u(t)=0).
\]
Their ``Independent Model'' approximates this distribution by
\[
P(X_{N(u)}(t)=x\mid X_u(t)=0)
\approx
\prod_{v\in N(u)}P(X_v(t)=x_v),
\]
which yields the same SWE recurrence above. They further introduce a
``Markov Model'', which retains the conditioning on the target node while
factorizing only the neighboring states,
\[
P(X_{N(u)}(t)=x\mid X_u(t)=0)
\approx
\prod_{v\in N(u)}
P(X_v(t)=x_v\mid X_u(t)=0).
\]
Thus, this method can be interpreted as an intermediate approximation
between the exact conditional formulation in
Equation~\ref{eq:exact_swe_marginals} and SWE: neighbor dependencies are
still factorized, but the dependence between a neighbor's activation and
the survival of \(u\) is retained.

The step-wise approximation in \cite{yang2012approximation} is closely
related but omits the survival term. Their nonlinear recurrence is
\[
\hat r_u(t)
=
1-\prod_{v\in N(u)}
(1-p_{vu}\hat r_v(t-1)),
\]
which corresponds to SWE without the correction for previous activation.
Under the additional assumption of small propagation probabilities they
linearize the product, obtaining
\[
\hat r_u(t)
\approx
\sum_{v\in N(u)}p_{vu}\hat r_v(t-1),
\]
which forms the basis of their GS approximation.

The earlier SPM and SP1M models \cite{SPM} can also be viewed from this
step-wise perspective. In SPM, activation is considered only at the
earliest possible time \(d(A,u)\), where the events ``active by time
\(d(A,u)\)'' and ``newly active at time \(d(A,u)\)'' coincide. Therefore
the noisy-OR update can be interpreted either cumulatively or step-wise.
SP1M additionally considers time \(d(A,u)+1\), where this equivalence no
longer holds, and consequently introduces an explicit survival correction
for the second activation opportunity.

The Aggregated Linear Effect model \cite{ale} (ALE) provides an especially simple linear approximation of the diffusion process. Let \(B\) denote the weighted adjacency matrix oriented such that \((Bx)_u=\sum_v p_{vu}x_v\). ALE estimates the cumulative activation vector after \(t\) propagation steps as

$$ \hat p(t)=\sum_{\tau=0}^{t}B^\tau p(0), $$

where \(B^0=I\). Each term \(B^\tau p(0)\) aggregates the contributions of weighted walks of length exactly \(\tau\). When \(\rho(B)<1\), the infinite series converges and

$$ \hat p(\infty)=(I-B)^{-1}p(0). $$

ALE is particularly attractive computationally because propagation is entirely linear. As a consequence of this linearization, ALE does not account for overlap between activation paths or repeated propagation through cycles, and its estimates are not intrinsically bounded by one.

\subsection{A Unified Mathematical View}
\label{subsec:unified-mathematical-view}

Although the approximation methods discussed above were introduced from
substantially different perspectives, many of them can be organized around
three modeling choices: (i) the temporal state that is propagated, (ii) the
amount of propagation history or context that is retained, and (iii) the rule
used to aggregate incoming activation attempts. This view exposes several
otherwise less apparent relationships between existing methods.

\subsubsection{A common propagation operator.}
Let \(G=(V,E)\) be a directed graph and let
\(\mathbf p^{(0)}\in[0,1]^{|V|}\) denote the vector of initial activation
probabilities. For each edge \(e=(u,v)\), let \(p_e=p_{uv}\). Define the
source-incidence matrix
\[
S\in\{0,1\}^{|E|\times |V|},
\qquad
S_{e,u}=1 \iff u=\operatorname{src}(e),
\]
the destination-incidence matrix
\[
D\in\{0,1\}^{|V|\times |E|},
\qquad
D_{v,e}=1 \iff v=\operatorname{dst}(e),
\]
and
\[
W=\operatorname{diag}(p_e).
\]
The weighted node-space propagation matrix is then
\[
B=DWS,
\]
such that
\[
(Bx)_v
=
\sum_{u\rightarrow v}p_{uv}x_u.
\]

The factorization \(B=DWS\) is introduced to expose the individual
edge-level contributions underlying the weighted adjacency operator.
While \(Bx\) directly sums all incoming contributions at each node,
nonlinear aggregation rules such as noisy-OR must be applied to each
incoming edge contribution before these contributions are combined.
The source-incidence matrix \(S\) therefore lifts a node-level state
\(x\) to edge space, \(W\) applies the corresponding edge probabilities,
and \(D\) subsequently aggregates edge-level quantities at their
destination nodes. In particular,
\[
B=DWS,
\]
so this representation leaves the usual linear propagation operator
unchanged while permitting nonlinear operations to be inserted before
the final destination-wise aggregation.

For a node $v$, suppose that $x_u$ denotes the activation quantity
propagated by an in-neighbor $u$. The corresponding incoming activation
attempt through edge $(u,v)$ has weight $p_{uv}x_u$. Under the usual
independence approximation, the probability that none of the incoming
neighbors activates $v$ is therefore
\[
\prod_{u\to v}(1-p_{uv}x_u),
\]
so that the noisy-OR aggregation is
\[
1-\prod_{u\to v}(1-p_{uv}x_u).
\]
This operation can be written simultaneously for all nodes using the
incidence matrices defined above. The vector $WSx$ contains the
edge-level quantities $p_{uv}x_u$, applying $\log$ converts the product
over incoming edges into a sum, and multiplication by $D$ performs this
sum separately for each destination node. Exponentiating then restores
the product. Hence we define
\begin{equation}
\label{eq:noisy_or_operator}
\mathcal G(x)
=
\mathbf 1-
\exp\!\left[
D\log(\mathbf 1-WSx)
\right],
\end{equation}
where the logarithm and exponential are applied componentwise. We take the convention $\log (0)=-\infty,~\exp(-\infty)=0$. Equivalently,
\[
\mathcal G(x)_v
=
1-\exp \left ( \sum_{u\to v} \log (1-p_{uv}x_u) \right )
=
1-\prod_{u\rightarrow v}(1-p_{uv}x_u).
\]
Thus, \(\mathcal G\) combines the incoming activation probabilities by the
usual noisy-OR rule. The differences between several existing approximations
can therefore be understood not through a different local propagation rule,
but through the state \(x\) to which this rule is applied.

\subsubsection{Cumulative and step-wise propagation.}
The Naive Method propagates cumulative activation probabilities. Writing
\(x_t\) for the estimated probability of being active by time \(t\),
\begin{equation}
\label{eq:nm_operator}
x_{t+1}
=
\mathbf p^{(0)}
+
(\mathbf 1-\mathbf p^{(0)})
\odot
\mathcal G(x_t).
\end{equation}

Step-wise estimation instead separates the probability of becoming newly
active from the cumulative probability. Let \(r_t\) denote the vector of
new-activation probabilities at exactly time \(t\), and let
\[
c_t=\sum_{\tau=0}^{t}r_\tau,
\qquad
r_0=c_0=\mathbf p^{(0)}.
\]
The canonical SWE update is
\begin{equation}
\label{eq:swe_operator}
r_t
=
(\mathbf 1-c_{t-1})
\odot
\mathcal G(r_{t-1}),
\qquad
c_t=c_{t-1}+r_t.
\end{equation}
The two formulations therefore use the same nonlinear incoming-attempt
operator, but propagate different temporal states. NM repeatedly propagates
a cumulative activation state, whereas SWE propagates only newly active
probability and explicitly accounts for depletion of nodes that have already
activated.

This distinction also clarifies the relation to several earlier methods.
SPM and SP1M \cite{SPM} restrict the admissible activation times to the
shortest-path layer, or to the shortest-path layer and one additional layer,
respectively. At the earliest possible activation time \(d(S,u)\), cumulative
and newly-active probabilities coincide, while the additional layer in SP1M
requires a correction for previous activation. The nonlinear step-wise
approximation of \cite{yang2012approximation} also propagates newly-active
quantities through the noisy-OR operator, but omits the depletion factor in
Equation~\ref{eq:swe_operator}; their GS approximation subsequently
linearizes this nonlinear update. The IC specialization of the unified model
in \cite{unified_influence_model} instead yields the complete SWE update,
including the survival term.

The distinction between cumulative and step-wise propagation is not merely
notational. As discussed in Section~\ref{subsec:exact-marginal}, exact
step-wise inference involves the conditional joint neighborhood law
\[
\Pr\!\left[F_u(R,t-1)\mid S_u(t-1)\right],
\]
where \(S_u(t-1)\) denotes survival of \(u\) through time \(t-1\).
The Independent Model of \cite{chen_taylor_2017} replaces this conditional
joint distribution by unconditional single-node marginals, resulting in the
same SWE-type recurrence. Their Markov approximation retains the conditioning
on the target node while factorizing the neighboring states, and can therefore
be viewed as an intermediate closure between the exact conditional
formulation and SWE.

\subsubsection{Linearization and the common core of NM and SWE.}
The local noisy-OR operator admits the expansion
\begin{equation}
\label{eq:noisy_or_linearization}
\mathcal G(x)_v
=
\sum_{u\rightarrow v}p_{uv}x_u
-
\sum_{\substack{u,w\rightarrow v\\u<w}}
p_{uv}p_{wv}x_ux_w
+\cdots,
\end{equation}
and consequently
\begin{equation}
\label{eq:G_linearization}
\mathcal G(x)=Bx+O(\|WSx\|^2).
\end{equation}
The matrix \(B\) is therefore the first-order propagation operator common to
both the cumulative and step-wise formulations.

This relationship can be made formal by introducing a low-activation scaling
\[
\mathbf p^{(0)}=\varepsilon s,
\qquad
\varepsilon\rightarrow0,
\]
and considering a fixed propagation horizon \(T\). For NM,
Equation~\ref{eq:nm_operator} gives
\[
x_{t+1}
=
\varepsilon s+Bx_t+O(\varepsilon^2),
\]
and hence
\begin{equation}
\label{eq:nm_to_ale}
x_T
=
\left(
I+B+B^2+\cdots+B^T
\right)\mathbf p^{(0)}
+
O(\varepsilon^2).
\end{equation}

The same first-order expression follows from SWE. Since
\(r_t,c_t=O(\varepsilon)\),
\[
r_t
=
(\mathbf 1-c_{t-1})
\odot
\left(Br_{t-1}+O(\varepsilon^2)\right)
=
Br_{t-1}+O(\varepsilon^2).
\]
Therefore
\[
r_t
=
B^t\mathbf p^{(0)}
+
O(\varepsilon^2),
\]
and summing the step-wise contributions yields
\begin{equation}
\label{eq:swe_to_ale}
c_T
=
\left(
I+B+B^2+\cdots+B^T
\right)\mathbf p^{(0)}
+
O(\varepsilon^2).
\end{equation}

Equations~\ref{eq:nm_to_ale} and \ref{eq:swe_to_ale} reveal a common
first-order core:
\begin{equation}
\label{eq:ale_common_core}
\boxed{
\hat p^{\mathrm{ALE}}(T)
=
\sum_{\tau=0}^{T}B^\tau\mathbf p^{(0)}.
}
\end{equation}
Thus, ALE \cite{ale} can be interpreted as the linearized low-activation
limit of either cumulative NM propagation or step-wise SWE propagation.
The two nonlinear formulations differ substantially away from this limit,
but those differences enter only through higher-order terms: overlap between
multiple incoming activation possibilities and, in SWE, depletion of nodes
that have already activated. This is why low activation regimes are not the 
best when studying the differences between approximation methods.

The walk interpretation of Equation~\ref{eq:ale_common_core} is also useful.
The quantity
\[
B^\tau\mathbf p^{(0)}
\]
sums the weights of directed walks of length exactly \(\tau\) originating
from the initial activation vector. ALE therefore replaces the probabilistic
union of overlapping activation routes by a linear weighted-walk expansion.
When the spectral radius of \(B\) is strictly below one, i.e. \(\rho(B)<1\), the corresponding infinite-horizon expression is
\[
\hat p^{\mathrm{ALE}}(\infty)
=
(I-B)^{-1}\mathbf p^{(0)}.
\]

\subsubsection{History and cavity context.}
The second major modeling choice is how much information about previous
propagation is retained. Node-level methods such as NM and SWE retain no
explicit information about the route by which activation probability reached
a node. DMP \cite{dmp,dmp_unidirectional} instead propagates edge-cavity
messages
\[
q_{u\rightarrow v}(t),
\]
thereby excluding the immediate contribution of \(v\) when computing the
activation of \(u\). Its local aggregation remains the same noisy-OR rule,
but it is applied in edge-cavity rather than node space.

More generally, the cavity hierarchy
\[
p_{u\rightarrow S}(t)
\]
retains an excluded set \(S\) and gives a systematic notion of increasing
history depth. First-order DMP retains only the immediately excluded node,
while higher-order cavity approximations retain additional local history.
This provides a natural hierarchy for controlling feedback that is absent
from ordinary node-level propagation.

In the linear regime, the same distinction has a particularly simple
interpretation. Let \(R\in\{0,1\}^{|E|\times |E|}\) denote the edge-reversal matrix, defined by
\[
R_{e,f}
=
\begin{cases}
1, & e=(u,v)\text{ and }f=(v,u),\\
0, & \text{otherwise}.
\end{cases}
\]
Thus, \(R_{e,f}=1\) precisely when \(f\) is the directed reverse of \(e\), and \(R_{e,f}=0\) otherwise.
Define
\[
C=SD-R,
\qquad
K=CW.
\]
Where powers of \(B\) enumerate weighted walks, powers of \(K\) enumerate
weighted nonbacktracking walks. For two directed edges \(f=(a,b)\) and \(e=(b,c)\), we have
\(C_{e,f}=1\) precisely when \(e\) can follow \(f\) without immediately
reversing it. Consequently,
\[
    K_{e,f}=C_{e,f}p_f.
\]
Thus, under the column-vector convention, \(K\) transfers edge-state
mass from \(f\) to each admissible successor \(e\), weighting the
transition by the probability of \(f\). In particular,
\[
    DWK^{L-1}Sx
\]
aggregates the contributions of all directed nonbacktracking walks of
length \(L\), assigning the walk
\[
    (e_0,\ldots,e_{L-1})
\]
the weight
\[
    x_{\operatorname{src}(e_0)}
    \prod_{j=0}^{L-1} p_{e_j}.
\] The linear cavity analogue of ALE therefore
replaces unrestricted walk propagation by nonbacktracking walk propagation.
This relationship is useful conceptually even when the corresponding
intermediate approximation is not attractive as a standalone algorithm.

Several additional hybrid methods arise naturally by modifying only one of
the choices above. For example, one may linearize the incoming aggregation
in SWE while retaining its depletion term, or linearize cavity propagation
while retaining edge-specific history. We use such methods primarily as
diagnostic ablations of the framework rather than as separate conceptual
families. As shown in Section~\ref{sec:results}, these intermediate variants
did not provide sufficiently favorable accuracy--runtime trade-offs in our
experiments to justify treating them as primary approximation methods.

\subsubsection{Structural restriction as an orthogonal strategy.}
Not all approximations operate by changing the propagated probability state.
A second class instead simplifies the set of admissible propagation routes.
MIA \cite{mia_icphard}, for example, restricts propagation to a
target-specific maximum influence in-arborescence, on which the noisy-OR
recursion can be evaluated efficiently. SPM and SP1M restrict admissible
path lengths, while neighborhood-based approximations such as
\cite{ale} restrict the explored local region. Such methods can be combined
conceptually with the three axes above: they first simplify the path space
and then select a temporal state, a history representation, and an incoming
aggregation rule on the resulting surrogate structure.

IRIE \cite{irie} occupies a somewhat different position. Its influence
estimation begins from cavity-like influence relations but collapses them to
node-level scores and introduces a damping parameter. Its resulting recurrence
is linear,
\[
r
=
\mathbf 1+\alpha Pr,
\]
so that
\[
r
=
(I-\alpha P)^{-1}\mathbf 1
=
\sum_{k\geq0}\alpha^kP^k\mathbf 1.
\]
IRIE can therefore also be interpreted as a damped weighted-walk
approximation, although its target quantity is influence score rather than
the node-level marginal vector.

\subsubsection{Taxonomy}
Table~\ref{tab:approximation_taxonomy} summarizes the resulting view.
The three principal axes are the temporal state being propagated, the amount
of history or contextual information retained, and the rule used to combine
incoming activation contributions. Path and topology restrictions are shown
within the history/context column because they modify the propagation space
rather than the local aggregation rule.

\begin{table*}[t]
\centering
\footnotesize 
\caption{Unified taxonomy of analytical Independent Cascade approximations.}
\label{tab:approximation_taxonomy}
\begin{tabular}{p{1.7cm}p{2.7cm}p{3.8cm}p{3.7cm}}
\toprule
\textbf{Method}
&
\textbf{Temporal state}
&
\textbf{History / context}
&
\textbf{Incoming aggregation}
\\
\midrule

NM
&
Cumulative marginal
&
Node state; no explicit path history
&
Nonlinear noisy-OR
\\ \midrule

SSS \cite{sss}
&
Cumulative steady-state marginal
&
Node state; no explicit path history
&
Nonlinear noisy-OR
\\ \midrule

SPM \cite{SPM}
&
Earliest activation layer
&
Shortest-path propagation only
&
Nonlinear noisy-OR
\\ \midrule

SP1M \cite{SPM}
&
Two activation layers
&
Shortest and shortest-plus-one path lengths
&
Nonlinear noisy-OR with second-layer survival correction
\\ \midrule

MIA \cite{mia_icphard}
&
Final / cumulative marginal
&
Target-specific maximum-influence in-arborescence
&
Nonlinear noisy-OR on the surrogate tree
\\ \midrule

\cite{yang2012approximation}
&
Step-wise activation
&
Node state; optionally combined with path restriction
&
Noisy-OR; GS further linearizes to an additive sum
\\ \midrule

\cite{unified_influence_model}
&
Step-wise and cumulative marginals
&
Node state; no explicit path history
&
Nonlinear noisy-OR with target depletion
\\ \midrule

Indep. \cite{chen_taylor_2017}
&
Step-wise susceptible / newly-active state
&
Unconditional node marginals
&
Nonlinear noisy-OR
\\ \midrule

Markov
\cite{chen_taylor_2017}
&
Step-wise susceptible / newly-active state
&
Target-conditioned neighbor marginals
&
Nonlinear noisy-OR
\\ \midrule

ALE \cite{ale}
&
Time-indexed walk contributions aggregated cumulatively
&
Node state; unrestricted weighted walks
&
Linear additive propagation
\\ \midrule

DMP \cite{dmp,dmp_unidirectional}
&
Cumulative marginal
&
First-order edge cavity / immediate exclusion
&
Nonlinear noisy-OR
\\ \midrule

$\text{DMP}_r,~ r>1$
&
Cumulative marginal
&
Higher-order cavity state / additional excluded history
&
Nonlinear noisy-OR
\\ \midrule

NBUB \cite{nbub}
&
Path-length-resolved contribution
&
Edge-specific nonbacktracking history
&
Nonlinear product bound
\\ \midrule

IRIE \cite{irie}
&
Time-aggregated influence score
&
Cavity relations collapsed to node-level scores
&
Damped linear additive propagation
\\ \midrule

TDA
\cite{burkholz2019,burkholz2021cascade}
&
Final-state marginal / cascade-size distribution
&
Tree or edge-cavity context
&
Noisy-OR for marginals; distributional convolution for cascade size
\\

\bottomrule
\end{tabular}
\end{table*}

The taxonomy emphasizes that these methods are not isolated constructions.
Many share the same local IC aggregation rule and differ primarily in what
state is propagated and what information about previous propagation is
discarded. NM, SWE, and DMP all retain the nonlinear noisy-OR operator but
use cumulative node states, step-wise node states, and edge-cavity states,
respectively. ALE and related linear methods arise by replacing that operator
with its first-order additive approximation. Methods such as MIA and SPM
instead act orthogonally by restricting the admissible propagation structure.
This view provides a common mathematical basis for comparing approximation
error, computational complexity, and the effect of increasingly rich
propagation context in the experiments that follow.

\section{New Approximation Methods}
\label{sec:new-methods}
A preliminary version of this work introduced the SWE-NO and modified ALE
estimators \cite{pozar2026approximation}. Their definitions are restated here
for completeness and to provide a common basis for the extensions developed
below. The present work further introduces a general continuation framework,
cavity-based variants, and higher-order dynamic message-passing corrections.

\subsection{Compensatory Continuation of Approximation Methods}
\label{subsec:Compensatory-Continuation}

Many analytical diffusion approximations exhibit systematic directional
errors.  Suppose that a primary approximation tends to be conservative, while
a second approximation tends to propagate activation more aggressively.  The
output of the former can then be used to initialize the latter.  Rather than
combining two final estimates algebraically, we therefore compose their
approximation operators.  We refer to this procedure as \emph{compensatory
continuation} (called smoothing in \cite{pozar2026approximation}).

Let $\Phi_F^{(T)}(p^{(0)})$ denote the node-level estimate produced by method
$F$ after $T$ propagation steps, and suppose that this estimate can be used as
the initial state of a second method $G$.  We define the $L$-step compensatory
continuation of $F$ by $G$ as
\[
\hat p^{(T,L)}
=
\Phi_G^{(L)}
\left(
\Phi_F^{(T)}(p^{(0)})
\right).
\]
We are implicitly assuming that the successive applications of the $F$ and $G$
methods keeps the output in the \([0,1]\) range, and take the convention $\Phi_G^{(0)}:=\operatorname{Id}$ to be the identity.
The continuation steps should not in general be interpreted as additional IC
time steps; they are additional approximation steps applied to the output of
the primary estimator.  The construction is most natural when the two methods
exhibit opposing systematic errors, although no improvement is guaranteed for
an arbitrary pair of estimators.


\subsection{SWE-NO}
\label{subsec:swe-no}

What was referred to simply as SWE in the preliminary version
\cite{pozar2026approximation} is here called SWE-NO (SWE with noisy-OR
accumulation) in order to distinguish it from the standard step-wise
estimator.  SWE-NO computes the approximate step-wise activation
quantities $\hat r_u(t)$ using the same recurrence as SWE, but changes how these quantities are
aggregated over time.  Instead of
\[
\hat p_u(t)=\sum_{i=0}^{t}\hat r_u(i),
\]
we use
\[
a_u(t)
=
1-\prod_{i=0}^{t}
\left(1-\hat r_u(i)\right).
\]
The quantity $a_u(t)$ should not be interpreted as the probability of
activation by time $t$, since the events represented by the
$\hat r_u(i)$ are intended to describe mutually exclusive activation
times.  Rather, the noisy-OR accumulation is a deliberate damping
mechanism.

For any $\hat r_u(i)\in[0,1]$,
\[
1-\prod_{i=0}^{t}(1-\hat r_u(i))
\leq
\sum_{i=0}^{t}\hat r_u(i),
\]
so SWE-NO is systematically no larger than SWE when applied to the
same step-wise estimates.  Moreover,
\[
1-\prod_i(1-\hat r_i)
=
\sum_i\hat r_i
+O(\|\hat r\|^2),
\]
and therefore the two aggregation rules agree to first order in weak
diffusion regimes.

Because SWE-NO deliberately damps the accumulated estimate, we use it
as a primary estimator within compensatory continuation and apply NM
as the continuation operator:
\[
\operatorname{SWE\text{-}NO+NM}(T,L)
=
\Phi_{\mathrm{NM}}^{(L)}
\left(a(T)\right).
\]
Here $a(T)$ is treated as the initial activation vector of the NM
continuation, however the $(1-p_u)$ term that multiplies the noisy-or update in the NM retains the original initial probabilities $p_u$ and not $a(T)$. These continuation iterations are approximation steps
and should not be interpreted as physical diffusion times following
time $T$.

For a fixed horizon $T$ and continuation depth $L$, the computation
requires
\[
O \big ((T+L)|E| \big )
\]
time and $O(|V|+|E|)$ memory under the usual sparse representation.


\subsection{Independent Propagation-Layer Approximation}
\label{subsec:IPL}
The Independent Propagation-Layer Approximation (IPL) is directly motivated by the simplicity of ALE \cite{ale}. In \cite{pozar2026approximation} it was called the modified aggregated linear effect, here we adopt a more descriptive name.

Let \(A_u^{(t)}\) be the event that node \(u\) is reachable by a length \(t\) activation path in the IC network. ALE approximates the probability mass associated with layer \(t\) by

\[
a(0)=p(0), \quad a(t+1)=B\cdot a(t) = B^{t+1}a(0)
\]
where \(B\in [0,1]^{|V|\times |V|}\) is the weighted adjacency matrix with \(B_{v,u}=p_{uv}\) for \((u,v)\in E\) and \(B_{v,u}=0\) otherwise. It therefore makes the approximation \(\Pr \left [A_u^{(t)} \right ]\approx a(t)\). ALE approximates the union of layer events as 
\[
\Pr \left [\cup_{t\le T}A_u^{(t)} \right ] \approx \sum_{t\le T}\Pr \left [A_u^{(t)} \right ] \approx \sum_{t\le T}a_u(t).
\]
Therefore we have
\[
p^{\text{ALE}}(T) = \sum_{t\le T}a(t).
\]

In a similar spirit to SWE-NO, IPL makes an additional approximation
when aggregating contributions across propagation depths: it treats the
layer events
\[
A_u^{(0)},A_u^{(1)},\ldots,A_u^{(T)}
\]
as mutually independent.
Consequently
\[
\Pr \left [ \cap_{t\le T}\overline{A_u^{(t)}} \right ] \approx \prod_{t\le T}\Pr\left [ \overline{A_u^{(t)}} \right ] \approx \prod_{t\le T}\big (1-a_u(t) \big ).
\]
This yields
\[
p^{\text{IPL}}_u(T) = 1-\prod_{t\le T}\big (1-a_u(t) \big )
\]
or in matrix form
\[
p^{\text{IPL}}(T) = 1-\bigodot_{t\le T}\big (1-B^tp(0) \big )
\]
where we take \(B^0=I\) to be the identity matrix.

ALE sums the contributions associated with walks of different lengths, thereby implicitly neglecting overlap between propagation depths. IPL retains ALE's linear propagation within each depth, but accumulates these depth-wise contributions through a survival product. Equivalently, it applies a noisy-OR approximation across propagation depths, preventing repeated influence arriving at different depths from being counted additively.

Expanding the temporal survival product gives
\[
1-\prod_{t=0}^{T}(1-a_u^{(t)})
=
\sum_t a_u^{(t)}
-
\sum_{s<t}a_u^{(s)}a_u^{(t)}
+\cdots.
\]
Thus standard ALE is precisely the first-order approximation of the
IPL aggregation rule.  IPL retains ALE's linear propagation within
each depth while restoring a nonlinear union approximation across
propagation depths.

The vectors $a^{(t)}$ need not be formed using explicit powers of
$B$.  Iterating $a^{(t+1)}=Ba^{(t)}$ requires $O(|E|)$ work per
layer, while the survival product requires $O(|V|)$ work.  Hence IPL
has time complexity $O(T(|E|+|V|))$, or $O(T|E|)$ for the sparse
graphs considered here, and requires $O(|E|+|V|)$ memory.

Note that in practical implementation of ALE and IPL clamping on each layer is used to prevent the computed values from exceeding $1$.


\subsection{Second-Order Incoming-Aggregation Corrections}
\label{subsec:second-order-ale}

Both ALE and IPL use the linear propagation operator $B$ to
approximate the probability that a node receives activation from its
incoming neighbors.  This corresponds to the first-order expansion of
the noisy-OR operator.  In particular, for
\[
y_{vu}=p_{vu}x_v,
\]
we have
\[
1-\prod_{v\in N(u)}(1-y_{vu})
=
\sum_{v\in N(u)}y_{vu}
-
\sum_{\substack{v,w\in N(u)\\v<w}}y_{vu}y_{wu}
+
O(\|y\|^3).
\]
Standard ALE retains only the first term.  We therefore consider a
second-order approximation that additionally accounts for pairwise
overlap between incoming activation opportunities.

Let
\[
\mathcal{G}_2(x)_u
=
\sum_{v\in N(u)}p_{vu}x_v
-
\sum_{\substack{v,w\in N(u)\\v<w}}
p_{vu}p_{wu}x_vx_w.
\]
Using
\[
\sum_{i<j}z_i z_j
=
\frac{1}{2}
\left[
\left(\sum_i z_i\right)^2-\sum_i z_i^2
\right],
\]
the second-order term can be computed without explicitly enumerating
pairs of incoming edges.  In the incidence-matrix notation introduced
previously, with $B=DWS$, this gives
\begin{equation}
\label{eq:noisy-or-second-order}
\mathcal{G}_2(x)
=
Bx
-
\frac{1}{2}
\left[
(Bx)^{\odot 2}
-
D(WSx)^{\odot 2}
\right],
\end{equation}
where $\odot$ denotes element-wise multiplication.

We define ALE2 by replacing the linear ALE propagation with this
second-order operator,
\[
a^{(0)}=p^{(0)},\qquad
a^{(t+1)}
=
\Pi_{[0,1]}
\left(
\mathcal{G}_2(a^{(t)})
\right),
\]
and accumulating the resulting propagation layers as in ALE,
\[
p^{\mathrm{ALE2}}(T)
=
\Pi_{[0,1]}
\left(
\sum_{t=0}^{T}a^{(t)}
\right),
\]
where $\Pi_{[0,1]}$ denotes component-wise clipping to the probability
interval.

The same correction can be combined with the independent
propagation-layer aggregation introduced above.  We refer to this
variant as IPL2:
\begin{equation}
\label{eq:ipl2}
p^{\mathrm{IPL2}}_u(T)
=
1-
\prod_{t=0}^{T}
\left(1-a_u^{(t)}\right),
\qquad
a^{(t+1)}
=
\Pi_{[0,1]}
\left(
\mathcal{G}_2(a^{(t)})
\right).
\end{equation}
Thus, ALE2 and IPL2 correct a different source of overlap from IPL
itself.  The second-order propagation operator accounts for pairwise
overlap between different incoming neighbors \emph{within} a
propagation layer, whereas IPL accounts for overlap between
contributions arriving at different propagation depths.

Despite including pairwise interactions, the second-order correction
does not require quadratic work in the node degrees.  Equation
\eqref{eq:noisy-or-second-order} can be evaluated using two sparse
edge aggregations and element-wise vector operations, requiring
$O(|E|+|V|)$ time per propagation step.  Consequently, both ALE2 and
IPL2 retain $O(T(|E|+|V|))$ time complexity, or $O(T|E|)$ on the
sparse networks considered here, with $O(|E|+|V|)$ memory.

Note that once clipping occurs, ALE2/IPL2 are no longer literal second-order Taylor approximations at that component; they are clipped heuristics based on that expansion. 

\subsection{Cavity Extensions of Existing Approximations}
\label{subsec:cavity-extensions}
We also considered cavity-based extensions of several existing approximations, including SWE, SWE-NO, ALE, IPL, and the steady-state DMP/SSS formulation, together with the hybrid SWE-HIB-C variant. These constructions use edge-specific cavity messages to suppress immediate feedback while preserving the characteristic update rule of the underlying method. Because these variants are secondary to the main developments, their full definitions and implementation-specific details are deferred to Appendix~\ref{app:cavity_extensions}.


\subsection{Higher-Order Dynamic Message Passing}
\label{subsec:higher-order-dmp}

Standard DMP retains only the immediately preceding cavity vertex.
Consequently, when the recursive expansion would require a larger
cavity set, all but the most recent exclusion are discarded.  This
suggests a hierarchy in which the last $h$ cavity vertices are retained
rather than only one. Hence, we propose to compute the cavity quantities using Equation~\ref{eq:independent_general_cavity} instead. As discussed in Section~\ref{subsec:exact-marginal}, retaining the
complete cavity history is generally intractable.  We therefore
truncate the history to its most recent $h$ vertices and refer to the
resulting approximation as DMP of order $h$.

For an ordered cavity history
$H=(v_1,\ldots,v_k)$, let $\tau_h(H,u)$ denote the suffix of
$(v_1,\ldots,v_k,u)$ containing at most the last $h$ vertices.  In
particular, $\tau_0(H,u)=\emptyset$.  We then define DMP of order $h$
by
\begin{equation}
\label{eq:dmp-order-h}
\hat p^{(h)}_{u\rightarrow H}(t)
=
1-(1-p_u)
\prod_{v\in N(u)\setminus H}
\left[
1-p_{vu}\,
\hat p^{(h)}_{v\rightarrow \tau_h(H,u)}(t-1)
\right],
\end{equation}
with
\[
\hat p^{(h)}_{u\rightarrow H}(0)=p_u.
\]
The corresponding node marginal is
\[
\hat p^{(h)}_u(t)
=
\hat p^{(h)}_{u\rightarrow\emptyset}(t).
\]
Thus, $h=0$ gives the NM recurrence, while $h=1$ gives standard DMP.

It was shown that the node marginals obtained by DMP are an upper bound on the true node marginals \cite{dmp}.
We prove the stronger statement that, for every full cavity history
$C$ and its retained suffix $H=\tau_h(C)$,
\[
p_{u\rightarrow C}(t)
\leq
\hat p^{(h)}_{u\rightarrow H}(t).
\]

We similarly prove by induction the stronger message-wise statement
\[
\hat p^{(h_1)}_{u\rightarrow H_1}(t)
\leq
\hat p^{(h_2)}_{u\rightarrow H_2}(t)
\]
whenever $h_1\geq h_2$ and $H_2$ is the corresponding suffix of
$H_1$.

\begin{theorem}[Monotone upper-bound hierarchy]
\label{thm:higher-order-dmp-bound}
Consider the standard Independent Cascade model, where all initial-seed
indicators and live-edge indicators are mutually independent. Deterministic
seeds are included as the degenerate case in which the corresponding seed
indicators are fixed.

Let $p_u(t)$ denote the exact probability that node $u$ is active by
time $t$, and let $\hat p^{(h)}_u(t)$ denote the order-$h$
approximation defined by Equation~\eqref{eq:dmp-order-h}. Then, for every
node $u$, every $t\geq 0$, and all integers $h_1\geq h_2\geq 0$,
\begin{equation}
\label{eq:dmp-order-hierarchy}
p_u(t)
\leq
\hat p^{(h_1)}_u(t)
\leq
\hat p^{(h_2)}_u(t).
\end{equation}
Consequently,
\[
p_u(t)
\leq \cdots
\leq
\hat p^{(2)}_u(t)
\leq
\hat p^{(1)}_u(t)
\leq
\hat p^{(0)}_u(t),
\]
where $\hat p^{(1)}$ is standard DMP and
$\hat p^{(0)}$ is NM.
\end{theorem}

Note that the theorem establishes monotonicity of the finite-order upper bounds.
It does not, by itself, imply that
$\hat p^{(h)}_u(t)\to p_u(t)$ as $h\to\infty$.
The proof of the above theorem can be found in Appendix~\ref{app:proof}.

In our empirical evaluation, we will explicitly consider $h=2$ alongside the rest of the methods.

For $h=2$, the state retains the two most recent vertices in the
cavity history. In addition to the standard first-order message
$\hat p_{u\rightarrow v}(t)$, the recursion may therefore require
second-order messages
\[
\hat p_{u\rightarrow v,w}(t),
\]
in which both $v$ and $w$ are excluded from the incoming neighborhood
of $u$. The additional exclusion of $w$ changes the update only when
$w\in N(u)$. Under the in-neighbor convention used here, a
second-order state reached along a directed history
\[
u\rightarrow v\rightarrow w
\]
is therefore nonredundant only when the closing edge
\[
w\rightarrow u
\]
is also present. In that case, the history forms a directed
three-cycle and the additional cavity exclusion removes a factor that
would otherwise reintroduce feedback through this cycle.

Let $\tau$ denote the number of such directed triangle-closing message
states. Importantly, $\tau$ counts oriented message configurations,
not unordered graph triangles. A single triple of vertices can
contribute more than one state when several directed orientations are
present.

The second-order implementation does not recompute a complete incoming
product separately for every triangle-closing state. For each
first-order context, the full incoming product is computed once and
cached. A nonredundant second-order state is then obtained by removing
the single factor corresponding to the additional excluded vertex.
When this factor is nonzero this can be done by division; equivalently,
one may use prefix--suffix products or log-domain products to avoid
numerical problems when a factor vanishes. Consequently, all
first-order states require $O(|E|)$ work per iteration, while all
nonredundant second-order corrections together require only
$O(\tau)$ additional work. The total running time is therefore
\[
O\!\left(T(|E|+\tau)\right),
\]
with
\[
O(|E|+\tau)
\]
memory.

The worst-case value of $\tau$ can be superlinear in the number of
edges, so the second-order method is not uniformly linear-time on
arbitrary sparse graphs. However, its overhead is controlled by the
number of triangle-closing configurations rather than by the total
number of length-two cavity histories. In particular, for
bounded-degree graph families,
\[
|E|=O(|V|)
\qquad\text{and}\qquad
\tau=O(|V|),
\]
so order-two DMP retains
\[
O(T|E|)
\]
asymptotic running time and differs from standard DMP only by a larger
constant factor. More generally, the same conclusion holds whenever
the number of directed triangle-closing states grows approximately
linearly with graph size.

Even when $\tau$ is larger, the triangle-selective representation can
be substantially smaller than a naive second-order construction. Such
a construction would maintain states for all length-two cavity
histories, whereas the present implementation instantiates only those
histories for which the additional exclusion actually changes the
recurrence.

The general idea of improving message-passing approximations by retaining
additional propagation history is not new and is closely related to several
higher-order and loop-corrected message-passing approaches in the literature.
Our contribution is the particular formulation of this idea for finite-horizon
IC marginals and the resulting monotone hierarchy of upper bounds. A more
detailed comparison with related higher-order constructions is given in
Appendix~\ref{app:higher-order-related-work}.


\section{Experiments}
\label{sec:experimental-methodology}

\subsection{Selection of Benchmark Methods}
\label{subsec:benchmark-methods}

The analytical methods discussed in Sections~\ref{sec:unified-view}--\ref{sec:new-methods} span
several distinct approximation strategies, and evaluating every known
variant would make the experimental comparison unnecessarily large.
We therefore selected a representative set intended to cover the main
approximation domains identified by the taxonomy in Section~\ref{sec:unified-view},
while also including the new methods introduced in this work. The aim
is consequently not to enumerate every published IC approximation, but
to compare representative methods that differ in their treatment of
temporal information, feedback, dependency, and admissible propagation
paths.

The cumulative mean-field family is represented by NM and
Steady-State Spread (SSS), respectively. NM provides the basic finite-horizon
node-level noisy-OR approximation, while SSS represents its
fixed-point counterpart. The step-wise family is represented by SWE,
together with cavity SWE, which introduces edge-specific messages to
suppress immediate feedback. We additionally evaluate SWE-NO and its
cavity version, which modify the temporal accumulation of the step-wise
estimates, and therefore test the compensatory-continuation ideas
introduced in this work.

SWE-HIB-C is included primarily as a diagnostic hybrid between
step-wise and cavity propagation, allowing us to examine whether the
two modeling choices, cumulative versus step-wise, provide complementary benefits.

The message-passing family is represented by standard first-order DMP
and by its fixed-point, effectively infinite-horizon variant. We
additionally evaluate second-order DMP. The latter is included primarily
as a proof of concept for the higher-order cavity hierarchy introduced
in Section~\ref{sec:new-methods}: rather than attempting to evaluate increasingly
large cavity orders, which rapidly increases the number of states,
order \(h=2\) provides the simplest nontrivial test of whether retaining
additional propagation history can improve over ordinary DMP.

Linear and walk-based approximations are represented by ALE and IPL.
The second-order variants ALE2 and IPL2 are additionally included in
the real-world benchmark to test whether correcting the leading
pairwise overlap neglected by the first-order linear propagation
improves estimation accuracy. The synthetic benchmark also contains
the additive SWE variant and its cavity analogue, which retain the
step-wise depletion mechanism while replacing the nonlinear incoming
aggregation by an additive approximation. Finally, SPM and SP1M are
included as representatives of methods that simplify diffusion by
restricting admissible propagation depths rather than modifying the
local IC update itself.

Several methods discussed in the preceding sections were considered, but
not retained for the full benchmark. SSS-Noself substantially increases
the computational cost of SSS, because it requires auxiliary
cavity-like fixed-point computations for individual target nodes. Its
cost was prohibitive at the graph sizes considered here. The
nonbacktracking upper bound (NBUB) was excluded for the same practical
reason: its additional edge-specific state and associated computation
made it substantially less scalable than the methods retained in the
main comparison.

MIA was also initially considered, particularly in combination with
compensatory continuation. Since MIA restricts diffusion to a maximum
influence in-arborescence and thereby discards admissible propagation
routes, its estimates are naturally biased downward, in fact, since its influence 
estimation is exact on the in-arborescence, it provides a lower bound. This initially
made it an appealing base estimator for compensatory continuation,
whose purpose is to recover influence omitted by an underestimating
first stage. In practice, however, the resulting MIA-based estimator
was less accurate than the strongest cavity and message-passing
methods, while construction of the target-specific influence
arborescences became computationally expensive in moderate and strong
diffusion regimes. We were therefore unable to evaluate MIA
consistently over the complete benchmark grid, and it is not included
in the principal aggregate comparison. Results on the subset of
instances for which the computation completed are reported separately.

\subsection{Experimental Setup}

The complete experimental methodology, including network generation and
preprocessing, diffusion-regime calibration, hyperparameter selection,
Monte Carlo reference estimation, evaluation metrics, implementation details and runtime
measurement, is provided in
Appendix~\ref{app:exp_method}. Here we summarize only the
main experimental design.

We evaluate the methods on seven real-world networks: NetHEPT,
WikiVote, Enron, Epinions, Slashdot, ca-CondMat, and cit-HepTh, spanning
collaboration, communication, trust, voting, and citation networks.
Their structural properties are summarized in
Table~\ref{tab:real-world-network-statistics}. We additionally consider
five synthetic network families, namely Erdős--Rényi,
Barabási--Albert, Holme--Kim, the Stochastic Block Model, and
Watts--Strogatz, generated across several density levels. Across both
benchmarks, we evaluate a broad range of diffusion conditions by varying
the seed-set size, edge-probability model, and overall diffusion
strength. In particular, the experiments include constant,
weighted-cascade, and heterogeneous Dirichlet edge probabilities,
multiple seed fractions, and weak, moderate, and strong diffusion
regimes. Hyperparameters are selected on separate calibration instances
and all reported comparisons are performed on held-out instances.

\begin{table*}[t]
\centering
\footnotesize
\caption{Structural properties of the real-world benchmark networks. For directed datasets, the edge count refers to directed arcs. Average degree, clustering, and triangle counts are computed on the simple undirected projection. Reciprocity is the fraction of directed arcs whose reverse arc is also present.}
\label{tab:real-world-network-statistics}
\begin{tabular}{lrrrrrrr}
\toprule
 & NetHEPT & WikiVote & Enron & Epinions & Slashdot & CondMat & cit-HepTh \\
\midrule
Nodes & 15\,229 & 7\,115 & 36\,692 & 75\,879 & 82\,168 & 23\,133 & 27\,770 \\
Edges & 31\,376 & 103\,689 & 367\,662 & 508\,837 & 870\,161 & 93\,439 & 352\,768 \\
Directed & False & True & True & True & True & False & True \\
Avg. degree & 4.12 & 28.32 & 10.02 & 10.69 & 12.27 & 8.08 & 25.37 \\
Avg. cl. co. & 0.499 & 0.141 & 0.497 & 0.138 & 0.060 & 0.633 & 0.312 \\
Triangles & 30\,846 & 608\,389 & 727\,044 & 1\,624\,481 & 602\,592 & 173\,361 & 1\,478\,735 \\
Reciprocity & -- & 0.056 & 1.000 & 0.405 & 0.841 & -- & 0.003 \\
\bottomrule
\end{tabular}
\end{table*}

\section{Results}
\label{sec:results}

We now present the empirical results of the benchmark introduced in
Section~\ref{sec:experimental-methodology}. Section~\ref{subsec:overall-results}
examines overall accuracy and the effect of the proposed higher-order
DMP correction. Section~\ref{subsec:stepwise-results} compares the
step-wise, cavity, and steady-state approximation families, while
Section~\ref{subsec:linear-results} considers the linear, additive, and
path-restricted methods. Finally,
Section~\ref{subsec:regime-results} analyzes how approximation accuracy
changes with diffusion strength, seed-set size, and edge-probability
model.

\subsection{Overall Accuracy and Higher-Order DMP}
\label{subsec:overall-results}

The main accuracy results for the real-world and synthetic benchmarks are
summarized in Figure~\ref{fig:pooled_rmse}, with additional metrics reported
in Tables~\ref{tab:real_world-pooled-metrics} and
\ref{tab:synthetic-pooled-metrics}. Across both benchmark suites,
\(\operatorname{DMP}_2\) achieved the lowest pooled median RMSE, followed by
standard first-order DMP. On the real-world benchmark, their pooled median
RMSE values were \(0.00728\) and \(0.00915\), respectively, while on the
synthetic benchmark they were \(0.00208\) and \(0.00211\).

The advantage of \(\operatorname{DMP}_2\) is also visible when the networks
are considered separately. It achieved the lowest dataset-level median RMSE
on each of the seven real-world networks and the lowest graph-family-level
median RMSE on each of the five synthetic graph families. This does not,
however, imply uniform dominance at the level of individual instances. On
the real-world benchmark, \(\operatorname{DMP}_2\) obtained the lowest RMSE
on 990 of the 1140 held-out instances and the lowest setting-level median
RMSE on 201 of the 228 feasible benchmark settings. The corresponding
numbers on the synthetic benchmark were 1571 of 2700 instances and 323 of
540 settings. Thus, the pooled advantage is particularly consistent on the
real-world networks, whereas several methods remain competitive with
\(\operatorname{DMP}_2\) on individual synthetic instances.

Here, \(\operatorname{DMP}_2\) denotes the second-order DMP approximation
introduced in Section~\ref{sec:new-methods}, while
\(\operatorname{DMP}_\infty\) denotes the steady-state version of
first-order DMP. The performance of \(\operatorname{DMP}_2\) therefore
provides empirical evidence that retaining additional cavity history can
correct errors that remain after immediate backtracking has been removed.

The improvement comes at an additional computational cost. On the
real-world benchmark, the median evaluation time increased from
approximately \(0.039\) seconds for DMP to \(0.566\) seconds for
\(\operatorname{DMP}_2\). The latter is comparable to the
\(0.513\)-second median runtime of \(\operatorname{DMP}_\infty\).
The overhead is considerably smaller on the synthetic networks, where the
corresponding median runtimes were \(0.0092\), \(0.0320\), and \(0.0825\)
seconds for DMP, \(\operatorname{DMP}_2\), and
\(\operatorname{DMP}_\infty\), respectively. As described in
Section~\ref{sec:experimental-methodology}, the topology-specific construction of the second-order
triangle states is performed separately from the timed evaluation, so these
numbers represent per-evaluation runtimes on a prepared graph.

The benefit of the second-order correction depends strongly on the graph
structure. This relationship is particularly clear in the synthetic
benchmark. For Watts--Strogatz graphs, the median RMSE decreased from
\(0.01450\) for DMP to \(0.01095\) for \(\operatorname{DMP}_2\), a
reduction of approximately \(24.5\%\). For Holme--Kim graphs, the
corresponding reduction was from \(0.00240\) to \(0.00217\), or
approximately \(9.5\%\). In contrast, the differences on Erdős--Rényi,
Barabási--Albert, and stochastic block model graphs were very small.
This is consistent with the interpretation of the second-order method:
its additional states become relevant primarily when short cyclic
structures create dependencies that are not removed by first-order
cavity messages.

The real-world networks show a qualitatively compatible, although less
conclusive, pattern. The dataset-level advantage of
\(\operatorname{DMP}_2\) exhibits a positive descriptive association
with average clustering, but only seven distinct real-world topologies
are available, so this observation should not be interpreted as a strong
statistical result. Moreover, the raw number of undirected triangles does
not show a corresponding positive relationship. This distinction is
important because the additional states of \(\operatorname{DMP}_2\)
arise specifically from directed triangle-closing configurations, whereas
the clustering coefficients and triangle counts reported in
Table~\ref{tab:real-world-network-statistics} are computed on the simple undirected projection.

The advantage of \(\operatorname{DMP}_2\) also depends on the diffusion
regime and seed fraction. When all real-world seed fractions are pooled,
the median reduction in RMSE relative to DMP is approximately
\(1.5\times10^{-5}\), \(0.00156\), and \(0.00308\) for
\(g^\star=0.10\), \(0.35\), and \(0.60\), respectively. Thus, in
absolute terms, the second-order correction becomes increasingly useful
as the diffusion becomes stronger.

There is, however, a notable interaction with the seed-set size. At the
smallest real-world seed fraction, \(s=0.01\), the median improvement is
largest in the moderate diffusion regime:
\[
0.00031,\qquad 0.00500,\qquad 0.00356
\]
for \(g^\star=0.10\), \(0.35\), and \(0.60\), respectively. The same
pattern is especially pronounced for the synthetic Watts--Strogatz
graphs, where the corresponding improvements are approximately
\(0.00492\), \(0.01494\), and \(0.00970\). At very small seed fractions, moderate diffusion can produce
the largest benefit from the second-order correction, whereas after
pooling over seed fractions the strongest diffusion regime yields the
largest overall improvement. This suggests that higher-order cavity
information becomes most useful when propagation is sufficiently strong
for short recurrent paths to matter, while the effect is reduced when
a large initial seed set already activates a substantial part of the
network.

\begin{figure}[htbp]
    \centering
    \begin{subfigure}[b]{0.48\textwidth}
        \centering
        \includegraphics[width=\textwidth]{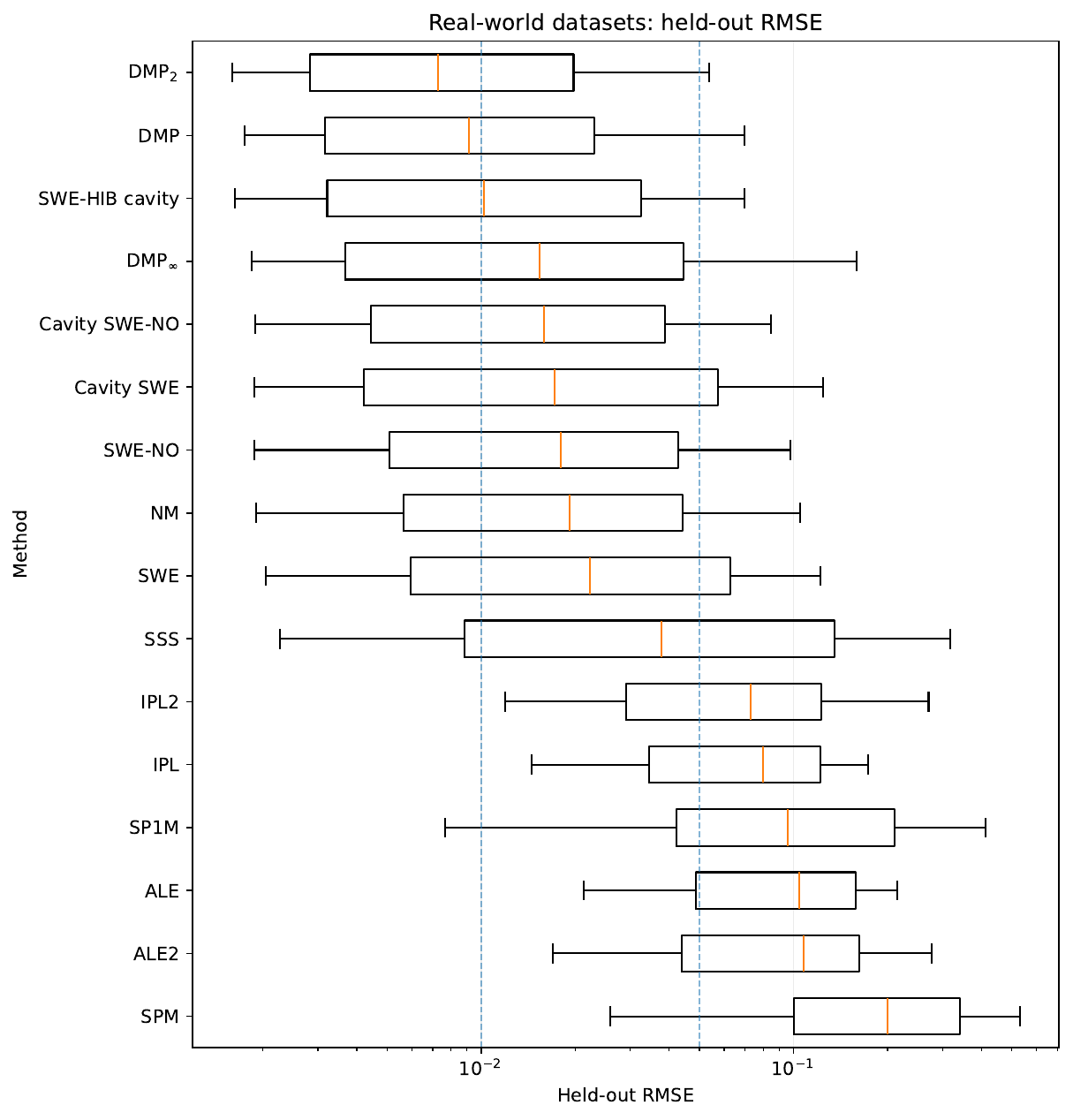}
        \caption{Real world benchmarks.}
        \label{fig:pdf1}
    \end{subfigure}
    \hfill 
    \begin{subfigure}[b]{0.48\textwidth}
        \centering
        \includegraphics[width=\textwidth]{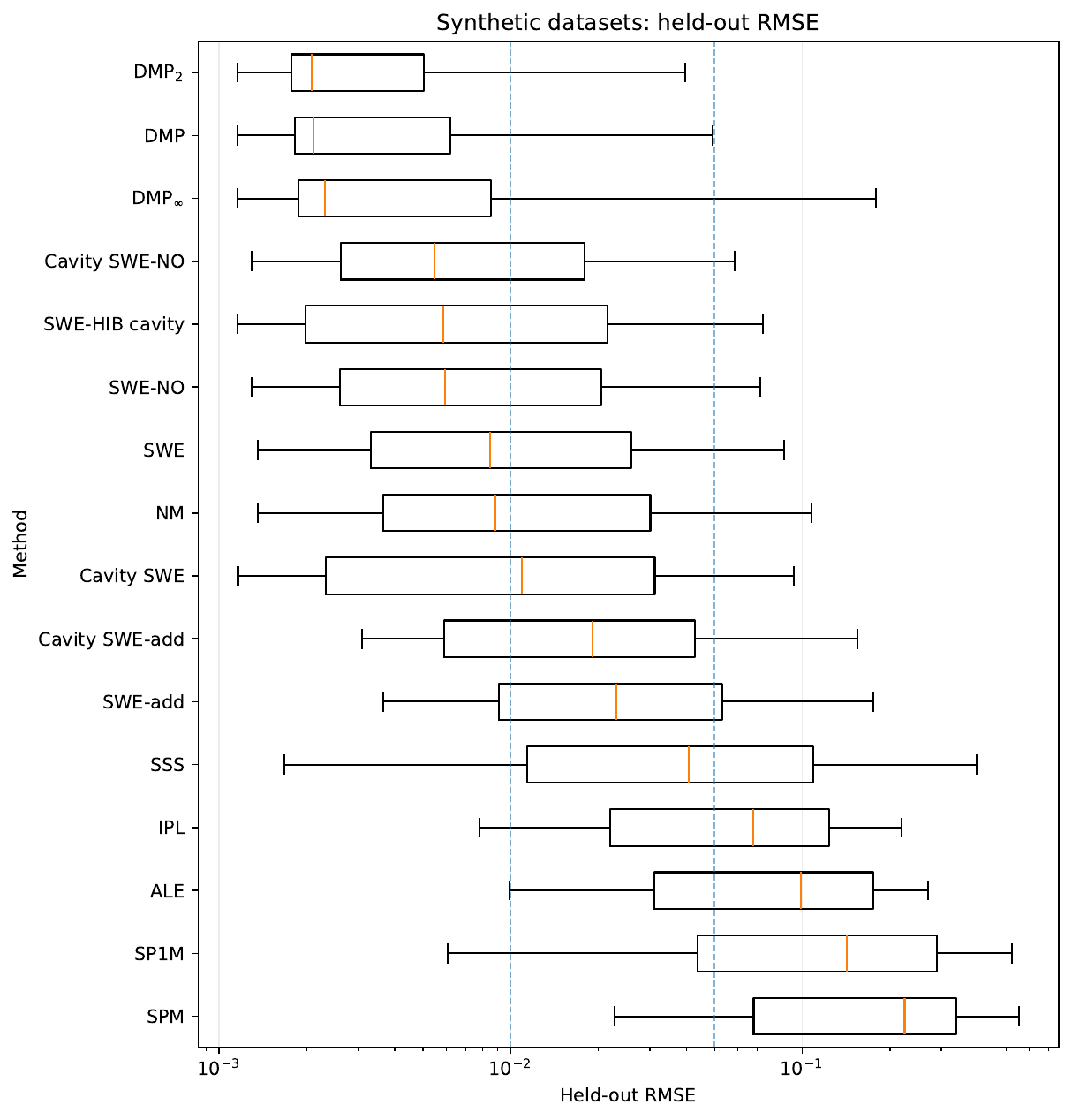}
        \caption{Synthetic benchmarks.}
        \label{fig:pdf2}
    \end{subfigure}
    
    \caption{Median RMSE accuracy of all methods.}
    \label{fig:pooled_rmse}
\end{figure}

\subsection{Step-Wise, Cavity, and Steady-State Approximations}
\label{subsec:stepwise-results}

Among the evaluated non-cavity methods, SWE-NO provides the strongest
overall accuracy. Its pooled median RMSE is \(0.01793\) on the real-world
benchmark and \(0.00596\) on the synthetic benchmark, compared with
\(0.02232\) and \(0.00849\) for standard SWE and \(0.01921\) and
\(0.00885\) for NM. On paired instances, SWE-NO outperforms standard SWE
on approximately \(91.8\%\) of real-world instances and \(93.1\%\) of
synthetic instances. It also outperforms NM on approximately \(80.9\%\)
and \(94.6\%\) of the corresponding instances. These results are
consistent with the empirical observations in
\cite{pozar2026approximation} and support the use of the temporal
noisy-OR accumulation as an effective damping mechanism for the
step-wise approximation.

The cavity extension of SWE-NO provides a further improvement in both
benchmark suites, reducing the pooled median RMSE to \(0.01592\) on the
real-world networks and \(0.00548\) on the synthetic networks. The
effect of cavity correction on standard SWE is less uniform. On the
real-world benchmark, cavity SWE reduces the pooled median RMSE from
\(0.02232\) to \(0.01714\). On the synthetic benchmark, however, its
pooled median RMSE is \(0.01089\), compared with \(0.00849\) for
standard SWE, despite cavity SWE outperforming SWE on a majority of
paired synthetic instances. This indicates that the cavity heuristic
frequently improves the step-wise estimate but can also produce
comparatively large errors in some regimes.

SWE-HIB-C performs particularly strongly. It has the third-lowest pooled
median RMSE on the real-world benchmark, \(0.01019\), and remains among
the strongest step-wise approximations on the synthetic benchmark, with
a median RMSE of \(0.00588\). The strong performance of this hybrid method indicates that approximation
quality depends not only on the state propagated through the network, but also
on how that state is converted into the final node marginals. SWE-HIB-C
combines step-wise cavity propagation with a cumulative DMP-style noisy-OR
readout, and the two mechanisms appear to provide complementary benefits.

This interpretation is also consistent with the different roles of
cavity messages in DMP and SWE. In DMP, an edge-specific message has a
direct interpretation as an approximation to diffusion in a cavity
graph. In the exact step-wise formulation, by contrast, the relevant
neighbor distribution is conditioned on the target node having remained
inactive. Deleting the target node is not equivalent to this
conditioning. The cavity-SWE variants therefore use cavity messages as
a heuristic approximation to the unavailable conditional information,
rather than as an exact reformulation of the step-wise dynamics.

The relative performance of the basic NM and SWE formulations is also
network-dependent. In the real-world benchmark, SWE obtains lower
dataset-level median RMSE on CondMat, Enron, and NetHEPT, while NM
performs better on Epinions, Slashdot, WikiVote, and cit-HepTh. In the
synthetic benchmark the differences are generally smaller: SWE has a
slightly lower pooled median RMSE than NM and performs particularly
well on Watts--Strogatz graphs, whereas the two methods are very close
on several other graph families. We therefore find no general dominance
relation between the basic cumulative and step-wise formulations.

The steady-state SSS approximation is considerably less accurate than
finite-horizon NM in both benchmark suites. Its pooled median RMSE is
\(0.03774\) on the real-world benchmark and \(0.04079\) on the
synthetic benchmark, compared with \(0.01921\) and \(0.00885\) for NM.
SSS has larger RMSE than NM on approximately \(99.2\%\) of paired
real-world instances and on every paired synthetic instance. This is
consistent with repeatedly applying a dependency-neglecting mean-field
update toward its fixed point, which can compound the overestimation
caused by recurrent and dependent propagation paths.

The steady-state DMP variant behaves more favorably than SSS but exhibits
the same general risk in difficult topologies. Its synthetic pooled
median RMSE of \(0.00231\) is close to those of DMP and
\(\operatorname{DMP}_2\), making it the third most accurate method by
this aggregate statistic. However, its error rises sharply on the
Watts--Strogatz family, where its median RMSE is \(0.06856\), compared
with \(0.01450\) for finite-horizon DMP. The pooled median therefore
conceals a substantial upper tail of difficult instances.

\begin{table}[t]
\caption{Pooled held-out performance on the Real-world datasets benchmark. Entries are medians over all successful held-out instance-level evaluations. Lower is better for RMSE, MAE, and runtime; higher is better for Spearman, Pearson, and AUC. The median runtime of MC simulations is $78.31$ s.}
\label{tab:real_world-pooled-metrics}
\begin{tabular}{lrrrrrl}
\toprule
Method & RMSE & MAE & Spearman & Pearson & AUC & Runtime (s) \\
\midrule
DMP$_2$ & 0.0073 & 0.0031 & 0.99606 & 0.99963 & 1.00000 & 0.5657 \\
DMP & 0.0092 & 0.0037 & 0.99511 & 0.99952 & 1.00000 & 0.03938 \\
SWE-HIB cavity & 0.0102 & 0.0050 & 0.99401 & 0.99936 & 0.99999 & 0.04449 \\
DMP$_{\infty}$ & 0.0153 & 0.0048 & 0.99423 & 0.99895 & 1.00000 & 0.5132 \\
Cavity SWE-NO & 0.0159 & 0.0063 & 0.99165 & 0.99881 & 0.99999 & 0.04645 \\
Cavity SWE & 0.0171 & 0.0097 & 0.99123 & 0.99858 & 0.99999 & 0.09963 \\
SWE-NO & 0.0179 & 0.0070 & 0.99066 & 0.99848 & 0.99999 & 0.03323 \\
NM & 0.0192 & 0.0075 & 0.99130 & 0.99826 & 1.00000 & 0.03017 \\
SWE & 0.0223 & 0.0121 & 0.98812 & 0.99740 & 0.99999 & 0.03454 \\
SSS & 0.0377 & 0.0158 & 0.98607 & 0.99552 & 0.99998 & 0.2383 \\
IPL2 & 0.0730 & 0.0399 & 0.96104 & 0.97765 & 0.99935 & 0.0134 \\
IPL & 0.0798 & 0.0436 & 0.95983 & 0.97752 & 0.99146 & 0.005469 \\
SP1M & 0.0958 & 0.0579 & 0.97283 & 0.96555 & 0.99994 & 0.1493 \\
ALE & 0.1040 & 0.0605 & 0.94423 & 0.96230 & 0.96084 & 0.005098 \\
ALE2 & 0.1080 & 0.0631 & 0.93586 & 0.94878 & 0.96861 & 0.01142 \\
SPM & 0.1997 & 0.1301 & 0.91555 & 0.87167 & 0.99885 & 0.1442 \\
\bottomrule
\end{tabular}
\end{table}

\subsection{Linear, Additive, and Path-Restricted Approximations}
\label{subsec:linear-results}

ALE and IPL are among the most computationally inexpensive methods in
both benchmark suites. On the real-world networks their median runtimes
are approximately \(0.00510\) and \(0.00547\) seconds, respectively,
while on the synthetic networks they are approximately \(0.00080\) and
\(0.00095\) seconds. These values are substantially below those of the
edge-cavity message-passing methods.

IPL is also consistently more accurate than ALE. Its pooled median RMSE
is \(0.07978\) compared with \(0.10401\) for ALE on the real-world
benchmark, and \(0.06804\) compared with \(0.09889\) on the synthetic
benchmark. At the instance level, IPL outperforms ALE on approximately
\(96.9\%\) of real-world instances and \(96.0\%\) of synthetic
instances. This provides strong empirical support for the motivation
of IPL in Section~\ref{sec:unified-view}: retaining the inexpensive
linear propagation of ALE while replacing additive accumulation across
propagation depths by a survival product corrects an important source
of repeated counting at essentially the same asymptotic cost.

The second-order incoming corrections ALE2 and IPL2 were evaluated in
the real-world benchmark. Their effect is not uniform. ALE2 obtains
lower RMSE than ALE on approximately \(73.8\%\) of paired instances,
yet its pooled median RMSE increases slightly from \(0.10401\) to
\(0.10798\), and its pooled mean RMSE also increases. Thus, ALE2
produces frequent small improvements but also a smaller number of
sufficiently large degradations to worsen the aggregate statistics.

IPL2 exhibits a more favorable median behavior. It outperforms IPL on
approximately \(74.0\%\) of paired instances and reduces the pooled
median RMSE from \(0.07978\) to \(0.07296\). Its pooled mean RMSE is,
however, higher than that of IPL, again indicating a heavier upper tail
of difficult instances. The second-order incoming correction therefore
cannot be interpreted as a uniform improvement. Its effect depends on
the instance, and the nonlinear propagation-depth aggregation used by
IPL appears to make more robust use of the correction than the additive
ALE aggregation. Importantly, IPL2 retains a median runtime of only
\(0.0134\) seconds and therefore remains among the most scalable methods
in the real-world comparison.

The synthetic benchmark additionally evaluates the additive SWE
variants. Standard additive SWE has a pooled median RMSE of \(0.02304\),
substantially above the \(0.00849\) obtained by nonlinear SWE. Its
cavity extension improves the median to \(0.01909\), but remains less
accurate than the corresponding nonlinear methods. Additive SWE is,
however, extremely inexpensive, with a median runtime of approximately
\(0.00123\) seconds. These results further illustrate the general
accuracy--cost trade-off associated with replacing nonlinear noisy-OR
aggregation by linear addition.

The path-restricted methods show a similar trade-off in a different
approximation domain. SP1M substantially improves upon SPM in both
benchmark suites, reducing the pooled median RMSE from \(0.19971\) to
\(0.09580\) on the real-world networks and from \(0.22446\) to
\(0.14191\) on the synthetic networks. Nevertheless, both remain less
accurate than the strongest node-level and cavity-based methods. This
suggests that restricting admissible propagation paths removes too much
diffusion information in the more demanding benchmark regimes, even
when the shortest-plus-one correction of SP1M is included.

\begin{table}[t]
\caption{Pooled held-out performance on the Synthetic datasets benchmark. Entries are medians over all successful held-out instance-level evaluations. Lower is better for RMSE, MAE, and runtime; higher is better for Spearman, Pearson, and AUC. The median runtime of MC simulations was $71.33$ s.}
\label{tab:synthetic-pooled-metrics}
\begin{tabular}{lrrrrrl}
\toprule
Method& RMSE & MAE & Spearman & Pearson & AUC & Runtime (s) \\
\midrule
DMP$_2$ & 0.0021 & 0.0016 & 0.99931 & 0.99992 & 1.00000 & 0.03201 \\
DMP & 0.0021 & 0.0016 & 0.99923 & 0.99991 & 1.00000 & 0.009151 \\
DMP$_{\infty}$ & 0.0023 & 0.0018 & 0.99924 & 0.99991 & 1.00000 & 0.08248 \\
Cavity SWE-NO & 0.0055 & 0.0039 & 0.99738 & 0.99952 & 1.00000 & 0.007849 \\
SWE-HIB cavity & 0.0059 & 0.0049 & 0.99871 & 0.99984 & 1.00000 & 0.01088 \\
SWE-NO & 0.0060 & 0.0044 & 0.99652 & 0.99944 & 1.00000 & 0.00477 \\
SWE & 0.0085 & 0.0063 & 0.99493 & 0.99904 & 1.00000 & 0.003953 \\
NM & 0.0088 & 0.0065 & 0.99484 & 0.99913 & 1.00000 & 0.003917 \\
Cavity SWE & 0.0109 & 0.0089 & 0.99813 & 0.99971 & 1.00000 & 0.02227 \\
Cavity SWE-add & 0.0191 & 0.0124 & 0.99277 & 0.99693 & 0.99995 & 0.008772 \\
SWE-add & 0.0230 & 0.0160 & 0.98291 & 0.99300 & 0.99989 & 0.001226 \\
SSS & 0.0408 & 0.0372 & 0.99085 & 0.99797 & 1.00000 & 0.02331 \\
IPL & 0.0680 & 0.0535 & 0.95801 & 0.96497 & 0.99967 & 0.0009482 \\
ALE & 0.0989 & 0.0742 & 0.94515 & 0.93441 & 0.99389 & 0.0008016 \\
SP1M & 0.1419 & 0.1210 & 0.93001 & 0.98665 & 1.00000 & 0.02044 \\
SPM & 0.2245 & 0.2014 & 0.84917 & 0.97750 & 1.00000 & 0.01982 \\
\bottomrule
\end{tabular}
\end{table}

\subsection{Dependence on Diffusion Regime}
\label{subsec:regime-results}

Diffusion strength is one of the clearest determinants of approximation
difficulty. For every finite-horizon method in both benchmark suites,
the pooled median RMSE increases from the weak regime
\(g^\star=0.10\) through \(g^\star=0.35\) to the strong regime
\(g^\star=0.60\). The two steady-state approximations constitute the
main exceptions to strict monotonicity in the synthetic benchmark.

The magnitude of this degradation differs substantially across method
families. On the real-world networks, the median RMSE of
\(\operatorname{DMP}_2\) increases from \(0.00238\) at
\(g^\star=0.10\) to \(0.01070\) and \(0.01773\) in the moderate and
strong regimes. On the synthetic networks the corresponding increase
is much smaller, from \(0.00128\) to \(0.00215\) and \(0.00221\).
In contrast, the linear methods deteriorate much more rapidly. For IPL,
the real-world medians are \(0.02570\), \(0.09559\), and \(0.12179\),
while the synthetic medians are \(0.01425\), \(0.07357\), and
\(0.13505\). ALE exhibits an even stronger increase, reaching a median
RMSE of \(0.19042\) in the strongest synthetic regime. These results
highlight that dependency and overlap corrections become increasingly
important as the cascade becomes more extensive.

Most nonlinear methods also become more accurate as the seed fraction
increases. For example, the median RMSE of \(\operatorname{DMP}_2\)
on the real-world benchmark decreases from \(0.01858\) at a \(1\%\)
seed fraction to \(0.00434\) at \(20\%\). On the synthetic benchmark
it decreases from \(0.00801\) to \(0.00187\). Similar trends are
observed for DMP, SWE, SWE-NO, NM, SSS, and the path-based methods.

This seed-fraction effect should nevertheless be interpreted with
caution. The reported RMSE includes the seed nodes themselves, whose
activation probabilities are deterministically equal to one and are
typically reproduced exactly by the approximation methods. Increasing
the seed fraction therefore mechanically increases the proportion of
nodes contributing essentially zero error. Moreover, ALE and IPL do
not exhibit a monotonic decrease in RMSE over all four seed fractions.
The observed seed-fraction dependence consequently cannot be attributed
solely to an intrinsically easier diffusion problem on the non-seed
nodes.

Finally, the edge-probability model also affects approximation
difficulty. When results are pooled over methods, the constant-edge
model yields the lowest median RMSE in both benchmark suites. On the
real-world benchmark, the pooled medians are approximately \(0.0210\),
\(0.0399\), and \(0.0429\) for constant, weighted-cascade, and
heterogeneous-Dirichlet probabilities, respectively. On the synthetic
benchmark they are approximately \(0.0122\), \(0.0130\), and \(0.0230\).
Thus, edge-weight heterogeneity generally makes the benchmark more
difficult.

This ordering is not universal across individual methods. In
particular, the ALE/IPL family sometimes performs as well as or better
under weighted-cascade or heterogeneous probabilities than under the
constant model. The effect of edge heterogeneity therefore interacts
with the approximation mechanism rather than imposing a uniform
difficulty ordering on every method.

\section{Discussion}
\label{sec:discussion}

We now interpret the empirical and theoretical results in the broader
context of analytical IC approximation. We first discuss the principal
sources of approximation error revealed by the benchmark and what the
results suggest about the roles of feedback, dependence, and retained
propagation information. We then consider how the higher-order cavity
construction may extend beyond the IC model. The following subsection
discusses differentiability and the practical problem of selecting
approximation hyperparameters, before we conclude with the main
limitations of the present study and directions for future work.

\subsection{Sources of Approximation Error}
\label{subsec:discussion-error}

The results in Section~\ref{sec:results} reinforce the importance of
feedback-loop correction in analytical approximations of the IC model.
The strong performance of DMP relative to node-level mean-field methods,
and the additional improvement obtained by \(\operatorname{DMP}_2\) in
clustered networks and selected diffusion regimes, indicate that repeated
propagation of information through short cycles is an important source of
approximation error. The synthetic experiments make this particularly
visible: the advantage of the second-order correction is small on graph
families with relatively little local clustering, but becomes substantially
larger on Watts--Strogatz and Holme--Kim networks.

Correcting feedback loops does not, however, remove all sources of
dependence. Even in a cavity graph, activation attempts arriving through
different neighbors of a node need not be independent. Two incoming
branches may share upstream ancestors, interact through longer cycles, or
be simultaneously affected by the same earlier activation events. Most
analytical IC approximations, including the message-passing methods studied
here, retain some form of conditional-independence approximation when
aggregating these incoming contributions. The remaining error after
short-loop correction therefore appears to be associated at least partly
with dependencies between distinct propagation branches.

Modeling such dependencies analytically is considerably more difficult than
suppressing explicit feedback through a bounded recent history. A richer
state representation could in principle retain information about joint
activation events or shared upstream structure, but this may lead to a
rapid increase in computational cost. Determining which dependencies are
most important in practice, and whether they can be represented without
sacrificing the near-linear scalability of message-passing methods, remains
an important direction for future work.

\subsection{Higher-Order Cavity Corrections Beyond the IC Model}
\label{subsec:discussion-higher-order}

Although the higher-order hierarchy developed in this work is derived
specifically for the IC model, the underlying construction is more general
than the noisy-OR update used here. At an abstract level, an ordinary
edge-cavity message is replaced by a message carrying a bounded recent
history, and all inputs corresponding to vertices retained in that history
are removed from the local update. The construction can therefore be viewed
as a finite-memory modification of a local cavity recursion.

An attractive feature of this formulation is that a higher-order state need
only be represented when the additional retained history actually changes
the local update. For the order-two IC construction, the additional
excluded vertex affects the recursion only when it is also an incoming
neighbor of the current node. The nonredundant second-order state space
therefore reduces to directed triangle-closing configurations, yielding
the complexity
\[
O\!\left(T(|E|+\tau)\right),
\]
where \(\tau\) denotes the number of such states. This is substantially
smaller than explicitly maintaining all possible length-two path histories
on many sparse networks.

The same bookkeeping principle could in principle be applied to other
cavity or belief-propagation equations. Whether it leads to an equally
useful approximation is a separate question. The IC model has several
properties that make the hierarchy particularly well behaved: diffusion is
progressive, the local update is monotone in the incoming messages, and
removing additional propagation routes cannot increase the true activation
probability. More general graphical models or recurrent dynamical systems
need not possess these properties, and a larger excluded history may not
produce a one-sided change in the resulting approximation.

This motivates a broader theoretical question. In the IC setting, the
higher-order construction yields a nested sequence of increasingly tight
upper approximations,
\[
p_u(t)
\le
\hat p^{(h+1)}_u(t)
\le
\hat p^{(h)}_u(t).
\]
It would be interesting to determine whether analogous monotonicity results
hold for higher-order cavity corrections in other probabilistic models, and
which structural properties are sufficient to guarantee such an ordering.
Monotonicity of the local update, positive dependence, and monotonicity
under removal of propagation routes are natural candidates. A
characterization of these conditions could provide a more general
theoretical basis for finite-history cavity hierarchies beyond progressive
diffusion.

More expressive loop-corrected and region-based message-passing methods can
represent dependencies that cannot be captured by a bounded path history,
for example correlations between incoming branches created by a larger
shared subgraph. Their additional expressive power is generally obtained
through substantially richer local state spaces. The finite-history
construction studied here therefore occupies a different point in the
accuracy--complexity trade-off: it targets a specific and easily identified
source of feedback while preserving a comparatively simple local recursion.
Extending the construction beyond order two, and characterizing the number
of nonredundant history states required at higher orders, are natural
directions for further work.

\subsection{Differentiability and Practical Parameter Selection}
\label{subsec:discussion-hyperparameters}

A further advantage of analytical diffusion approximations is that several
of the methods considered here are differentiable with respect to the
diffusion parameters. This makes them suitable for settings in which
diffusion estimation is embedded within a larger learning procedure, such
as edge-probability estimation, inverse diffusion, or network-parameter
learning. In contrast to Monte Carlo simulation, the approximation itself
can be represented as a deterministic computational graph and differentiated
using standard automatic-differentiation techniques. We do not investigate
such downstream learning tasks in the present work, but they constitute an
important motivation for computationally inexpensive analytical
approximations.

Practical use of these methods also depends critically on hyperparameter
selection. For most finite-horizon methods the main parameter is the
iteration or propagation horizon \(T\), while compensatory continuation,
introduced in Section~\ref{subsec:Compensatory-Continuation}, additionally
requires the continuation depth \(L\). The benchmark demonstrates that
relatively small changes in these parameters can produce appreciable
differences in both accuracy and runtime. At the same time, the strongest
message-passing methods, particularly the DMP variants, appear less
sensitive to small perturbations of the selected horizon than several of
the simpler approximations.

Our benchmark deliberately gives each method access to calibration
instances drawn from the same experimental setting as its held-out test
instances. The selected parameters should therefore be interpreted as
representing favorable setting-specific calibration rather than a
deployment scenario in which the optimal horizon is known in advance.
In practical applications, parameter selection should preferably be
performed on calibration instances that resemble the intended downstream
setting in graph topology, seed-set size, edge-probability distribution,
and diffusion strength.

We also conducted preliminary experiments aimed at predicting suitable
horizons directly from instance characteristics, including topological
features, seed-set size, secondary activation, and edge-probability
statistics. Simple linear models were unable to predict the selected
horizon reliably. Decision-tree models improved predictive performance,
suggesting a substantially nonlinear relationship between the instance
characteristics and the preferred approximation depth, but the resulting
accuracy was not sufficient to justify inclusion as a standalone result.

A more systematic treatment of automatic parameter selection therefore
remains an interesting direction for future work. One possibility is to
learn a predictor of suitable horizons from graph- and diffusion-level
features. Differentiable approximations also make softer formulations
possible, for example by replacing a hard integer horizon with learnable
weights or gates over successive propagation layers. Such a relaxation
could allow the effective propagation depth to be optimized jointly with
other model parameters using a diffusion-prediction loss. Whether this can
generalize across graph families and diffusion regimes remains an open
question.

\subsection{Limitations and Future Work}
\label{subsec:limitations}

The empirical evaluation covers a broad collection of real-world and
synthetic networks, edge-probability models, seed-set sizes, and diffusion
strengths, but it is necessarily finite. The pooled summaries used
throughout the paper are useful for comparing the overall behavior of the
methods, but they can hide substantial regime-specific variation. The
interaction observed between seed fraction and diffusion strength for
\(\operatorname{DMP}_2\), for example, would be obscured by considering
only a single pooled statistic. 

The reference activation probabilities are estimated using finite Monte
Carlo simulations and consequently contain sampling error. The number of
simulations used in the benchmark makes these estimates substantially more
accurate than the analytical approximations in most regimes, but very small
differences between strong methods should nevertheless be interpreted with
care. Future studies could complement Monte Carlo evaluation with graph
families for which exact IC marginals can be computed, either through
restricted topology, exact dynamic programming, or exhaustive enumeration
on sufficiently small networks.

A related limitation is that the reported error metrics include the seed
nodes. Since seed nodes have deterministic activation probability one and
are generally reproduced exactly by the approximations, increasing the
seed fraction mechanically increases the proportion of nodes with
essentially zero prediction error. Consequently, part of the observed
improvement in all-node RMSE at larger seed fractions does not necessarily
reflect improved prediction of the subsequent diffusion process. Reporting
the corresponding metrics restricted to non-seed nodes would provide a
useful complementary analysis.

The benchmark also relies on setting-specific hyperparameter calibration
using Monte Carlo reference estimates. Such supervision may not be
available in a practical deployment setting, and the reported performance
should therefore not be interpreted as requiring no prior calibration.
Developing robust parameter-selection rules that transfer between related
networks and diffusion regimes remains an important practical problem.

Our analytical and empirical treatment is restricted to the Independent
Cascade model. Several of the underlying ideas, including compensatory
continuation, finite-history cavity corrections, and differentiable
analytical propagation, may extend to other progressive diffusion models,
but their validity and accuracy outside IC have not been established here.

We additionally evaluate the methods only as estimators of node activation
probabilities. We do not test their performance within downstream
optimization procedures such as influence maximization. High accuracy of
node marginals, or even high rank correlation between those marginals, does
not by itself guarantee accurate marginal-gain estimates or high-quality
seed selection. Evaluating the approximations within downstream
optimization and learning tasks therefore remains separate future work.

Finally, the improved accuracy of \(\operatorname{DMP}_2\) comes with a
substantial computational overhead relative to first-order DMP. The method
remained scalable on all benchmark instances, but its cost depends on the
number of directed triangle-closing states and can become large in highly
clustered networks. Moreover, the reported per-evaluation runtime excludes
the topology-specific construction of these second-order states. The method
should therefore be viewed as an accuracy-oriented point in the
accuracy--runtime trade-off rather than as a uniformly preferable
replacement for first-order DMP. Quantifying when the additional
second-order state is worth its cost, and developing more compact
representations of higher-order cavity information, are natural directions
for further work.

\section{Conclusion}
\label{sec:conclusion}

In this work, we provide a broad analysis of analytical approximation
methods for the Independent Cascade model. We organize existing approaches
within a common framework based on three principal dimensions: the temporal
state being propagated, the information retained by the approximation, and
the rule used to aggregate incoming influence. This unified view clarifies
the relationships between several previously distinct approximation
families and motivates a number of new constructions and intermediate
variants. Together, these methods provide a finer spectrum of
accuracy--runtime trade-offs than is apparent when the existing
approximations are considered in isolation.

We complement the theoretical analysis with an extensive empirical
comparison focused on the estimation of node-level marginal activation
probabilities. The benchmark covers both real-world and synthetic networks
and varies graph topology, seed-set size, edge-probability model, and
diffusion strength. The results demonstrate substantial differences between
approximation families and show that their relative performance depends not
only on computational complexity, but also on the structural and dynamical
regime in which they are applied.

To our knowledge, this work constitutes the first comprehensive empirical study of analytical IC approximation methods. We evaluate a broad range of approximation families under a common protocol for node-level marginal activation probabilities across diverse real-world and synthetic diffusion regimes.

Among the evaluated methods, the proposed second-order extension of dynamic
message passing, \(\operatorname{DMP}_2\), achieved the strongest overall
accuracy. It obtained the lowest pooled median RMSE in both benchmark suites,
as well as the lowest dataset-level median RMSE on every real-world network
and the lowest graph-family-level median RMSE on every synthetic family.
Its advantage over first-order DMP is particularly pronounced in regimes
where short-loop effects are important, including highly clustered
synthetic networks and selected low-seed, moderate-to-strong diffusion
settings. This improvement comes with a substantial computational overhead,
but the method remained scalable on all networks considered.

At the opposite end of the accuracy--runtime spectrum, the proposed
Independent Propagation-Layer approximation (IPL) substantially improves
upon ALE while retaining the same \(O(T|E|)\) asymptotic complexity on
sparse networks. The second-order incoming-aggregation correction further
improves IPL on the majority of the real-world benchmark instances, although
its benefit is not uniform. These results illustrate that relatively small
changes to how propagation contributions are aggregated can produce
meaningful accuracy improvements without sacrificing scalability.

More broadly, the experiments reinforce the importance of explicitly
accounting for feedback and dependence in analytical diffusion
approximations. The higher-order DMP hierarchy provides one systematic way
of correcting progressively richer short-loop effects, while the unified
framework exposes additional opportunities for combining temporal,
structural, and aggregation corrections. At the same time, the benchmark
shows that no single approximation mechanism dominates every regime or
every point on the computational trade-off.

Overall, the results suggest that analytical IC approximation should not be
viewed as a choice between a small number of isolated methods, but as a
design space in which retained information, temporal representation, and
aggregation rules can be varied according to the requirements of the
application. We hope that the theoretical framework, new approximation
methods, and benchmark presented here provide a useful basis for developing
more accurate, scalable, and application-specific diffusion estimators.

\backmatter

\section*{Declarations}

\subsection*{Availability of data and materials}
The real-world datasets analysed in this study are publicly available from their original public sources, which are listed in the project repository. WikiVote, Enron, Epinions, Slashdot, CondMat, and cit-HepTh are available from the Stanford Network Analysis Project (SNAP), while NetHEPT is available from the IMPC Github repository. The datasets are cited in the manuscript according to the citation recommendations provided by their respective repositories.

Synthetic datasets and experimental results were generated using the code provided in the public project repository. Due to their size, the generated experimental result files are not included in the repository; however, they can be reproduced using the provided code and experimental configuration.\\

\begin{tabularx}{\textwidth}{@{}l X@{}}
\textbf{Project name:} & \texttt{Independent\_Cascade\_Approximations} \\
\textbf{Project home page:} & \url{https://github.com/Matic054/Independent_Cascade_Approximations} \\
\textbf{Archived version:} & Git commit \texttt{c7bfba7bd7053f8c28aa919830629ca2723262e3} \\
\textbf{Operating system:} & Ubuntu 24.04 LTS \\
\textbf{Programming language:} & Python 3 \\
\textbf{Other requirements:} & Required Python dependencies are listed in the repository \\
\textbf{License:} & MIT \\
\textbf{Restrictions on use by non-academics:} & None \\
\end{tabularx}

\subsection*{Use of artificial intelligence}
GPT-5.6 was used to assist with code implementation and refactoring,
language editing, manuscript organization, and discussion of methodological
and presentational alternatives. The authors specified the research
questions, algorithms, experimental methodology, analyses, and scientific
interpretation, and independently reviewed and verified all AI-assisted
content, code, citations, and experimental outputs. The authors take full
responsibility for the content of the manuscript and the reported results.

\subsection*{Competing interests}
The authors declares that he has no competing interests.

\subsection*{Funding}
Miklós Krész has been supported by the Slovenian Research and Innovation Agency (ARIS) through grants J7-60122, N2-0486, J1-70046, BI-HU/26-27-006 and BI-HU/26-27-007. He is also grateful for the support of the research program CogniCom (0013103) at the University of Primorska. 

\subsection*{Authors' contributions}
MP conceived the study, developed the methods, implemented the experiments,
analysed the results, and drafted the manuscript. MK supervised the study,
contributed to the development and interpretation of the methodology, and
reviewed and edited the manuscript. Both authors read and approved the final
manuscript.

\begin{appendices}

\section{Proof of Theorem~\ref{thm:higher-order-dmp-bound}}\label{app:proof}

\noindent\textbf{Theorem~\ref{thm:higher-order-dmp-bound} (restated).}
Assume the standard Independent Cascade model, where all initial-seed
indicators and live-edge indicators are mutually independent. Let
$p_u(t)$ denote the exact probability that node $u$ is active by time
$t$, and let $\hat p_u^{(h)}(t)$ denote the order-$h$ approximation.
Then, for every node $u$, every $t\geq 0$, and all integers
$h_1\geq h_2\geq 0$,
\[
p_u(t)
\leq
\hat p_u^{(h_1)}(t)
\leq
\hat p_u^{(h_2)}(t).
\]
Consequently,
\[
p_u(t)
\leq \cdots \leq
\hat p_u^{(2)}(t)
\leq
\hat p_u^{(1)}(t)
\leq
\hat p_u^{(0)}(t).
\]

\begin{proof}
A cavity history is an ordered tuple
\[
C=(c_1,\ldots,c_k),
\]
and we write
\[
[C]=\{c_1,\ldots,c_k\}
\]
for the underlying set of vertices occurring in the history. For a node
$u\notin[C]$, let
\[
C\oplus u=(c_1,\ldots,c_k,u)
\]
denote history concatenation. Thus
\[
[C\oplus u]=[C]\cup\{u\}.
\]

Let $p_{u\rightarrow C}(t)$ denote the exact activation probability of
$u$ by time $t$ in the cavity graph
\[
G\setminus[C].
\]
We first derive a one-step upper bound for this exact cavity
probability.

For each
\[
v\in N^{\mathrm{in}}(u)\setminus[C],
\]
let
\[
X_v(t-1)
=
\mathbf 1
\left\{
v\text{ is active by time }t-1
\text{ in }G\setminus[C\oplus u]
\right\},
\]
and let $L_{vu}$ be the indicator that the edge $v\to u$ is live.
Define
\[
Y_v(t)=X_v(t-1)L_{vu}.
\]
Thus $Y_v(t)=1$ precisely when the incoming branch through $v$ can
activate $u$ by time $t$.

The exact cavity probability can therefore be written as
\begin{align}
p_{u\rightarrow C}(t)
&=
p_u+
(1-p_u)
\Pr\left(
\bigcup_{v\in N^{\mathrm{in}}(u)\setminus[C]}
\{Y_v(t)=1\}
\right)
\nonumber\\
&=
1-(1-p_u)
\Pr\left(
\bigcap_{v\in N^{\mathrm{in}}(u)\setminus[C]}
\{Y_v(t)=0\}
\right).
\label{eq:cavity-exact-union}
\end{align}
Equivalently, if $R$ denotes the random set of active incoming
neighbors, Equation~\eqref{eq:cavity-exact-union} is the cavity analogue
of the subset expansion
\[
p_u+
(1-p_u)
\sum_R
p_{R\rightarrow u}(t-1)
\left[
1-\prod_{v\in R}(1-p_{vu})
\right].
\]

Each event $\{Y_v(t)=1\}$ is increasing in the mutually independent
seed and live-edge indicators: changing a seed indicator from inactive
to active, or an edge indicator from blocked to live, cannot destroy an
existing activation path. Hence each event $\{Y_v(t)=0\}$ is
decreasing. By Harris' inequality for product measures
\cite{harris1960lower}, decreasing events are positively associated.
Therefore,
\begin{equation}
\Pr\left(
\bigcap_{v\in N^{\mathrm{in}}(u)\setminus[C]}
\{Y_v(t)=0\}
\right)
\geq
\prod_{v\in N^{\mathrm{in}}(u)\setminus[C]}
\Pr\left(Y_v(t)=0\right).
\label{eq:harris-cavity}
\end{equation}

Moreover, since $X_v(t-1)$ is defined in
$G\setminus[C\oplus u]$, it does not depend on the live-edge indicator
$L_{vu}$. Hence
\[
\Pr(Y_v(t)=1)
=
p_{vu}\,
p_{v\rightarrow C\oplus u}(t-1).
\]
Substituting this into
Equation~\eqref{eq:harris-cavity} and then into
Equation~\eqref{eq:cavity-exact-union} yields
\begin{equation}
\label{eq:exact-cavity-upper-step}
p_{u\rightarrow C}(t)
\leq
1-(1-p_u)
\prod_{v\in N^{\mathrm{in}}(u)\setminus[C]}
\left[
1-p_{vu}\,
p_{v\rightarrow C\oplus u}(t-1)
\right].
\end{equation}
Thus the usual noisy-OR cavity update is obtained by replacing the
positively dependent incoming branches by independent branches with the
same cavity marginals. Positive association implies that this
factorization underestimates the probability that all incoming branches
fail, and therefore overestimates the activation probability.

We now prove a stronger statement than the theorem. For any ordered
history $C$, define
\[
H_h(C)=\tau_h(C),
\]
where $\tau_h(C)$ is the suffix of $C$ containing at most its $h$ most
recent vertices. We claim that, for every history $C$ and every finite
order $h$,
\begin{equation}
\label{eq:full-history-upper}
p_{u\rightarrow C}(t)
\leq
\hat p^{(h)}_{u\rightarrow H_h(C)}(t).
\end{equation}

At $t=0$, both sides equal $p_u$. Suppose that
Equation~\eqref{eq:full-history-upper} holds at time $t-1$ for every
valid cavity history. Applying
Equation~\eqref{eq:exact-cavity-upper-step} gives
\[
p_{u\rightarrow C}(t)
\leq
1-(1-p_u)
\prod_{v\in N^{\mathrm{in}}(u)\setminus[C]}
\left[
1-p_{vu}\,
p_{v\rightarrow C\oplus u}(t-1)
\right].
\]
By the induction hypothesis, applied to node $v$ and the full cavity
history $C\oplus u$,
\[
p_{v\rightarrow C\oplus u}(t-1)
\leq
\hat p^{(h)}_{v\rightarrow H_h(C\oplus u)}(t-1).
\]
Since
\[
H_h(C\oplus u)
=
\tau_h(C\oplus u)
=
\tau_h(H_h(C)\oplus u),
\]
the incoming message on the right is exactly the order-$h$ message
required by the truncated recursion.

The noisy-OR expression
\[
1-(1-p_u)\prod_v(1-p_{vu}x_v)
\]
is coordinatewise non-decreasing in the variables $x_v$. Therefore,
\begin{align}
p_{u\rightarrow C}(t)
&\leq
1-(1-p_u)
\prod_{v\in N^{\mathrm{in}}(u)\setminus[C]}
\left[
1-p_{vu}\,
\hat p^{(h)}_{v\rightarrow H_h(C\oplus u)}(t-1)
\right].
\label{eq:upper-before-truncation}
\end{align}

Since $H_h(C)$ is a suffix of $C$,
\[
[H_h(C)]\subseteq[C],
\]
and hence
\[
N^{\mathrm{in}}(u)\setminus[C]
\subseteq
N^{\mathrm{in}}(u)\setminus[H_h(C)].
\]
Every factor in the survival product lies in $[0,1]$. Consequently,
adding the neighbors forgotten by the truncated history can only
decrease the product:
\[
\prod_{v\in N^{\mathrm{in}}(u)\setminus[H_h(C)]}
\left[
1-p_{vu}\,
\hat p^{(h)}_{v\rightarrow H_h(C\oplus u)}(t-1)
\right]
\leq
\prod_{v\in N^{\mathrm{in}}(u)\setminus[C]}
\left[
1-p_{vu}\,
\hat p^{(h)}_{v\rightarrow H_h(C\oplus u)}(t-1)
\right].
\]
Combining this inequality with
Equation~\eqref{eq:upper-before-truncation} gives
\[
p_{u\rightarrow C}(t)
\leq
1-(1-p_u)
\prod_{v\in N^{\mathrm{in}}(u)\setminus[H_h(C)]}
\left[
1-p_{vu}\,
\hat p^{(h)}_{v\rightarrow H_h(C\oplus u)}(t-1)
\right].
\]
The right-hand side is precisely
\[
\hat p^{(h)}_{u\rightarrow H_h(C)}(t),
\]
which proves Equation~\eqref{eq:full-history-upper}. Taking the empty
history $C=()$ yields
\[
p_u(t)\leq\hat p^{(h)}_u(t)
\]
for every finite $h$.

It remains to compare two approximation orders. Let
$h_1\geq h_2$, and for a common full history $C$ define
\[
H_1=H_{h_1}(C),
\qquad
H_2=H_{h_2}(C).
\]
Because $H_2$ is a suffix of $H_1$,
\[
[H_2]\subseteq[H_1].
\]
We prove by induction on $t$ that
\begin{equation}
\label{eq:history-order-monotonicity}
\hat p^{(h_1)}_{u\rightarrow H_1}(t)
\leq
\hat p^{(h_2)}_{u\rightarrow H_2}(t).
\end{equation}

The claim is immediate at $t=0$. Suppose that it holds at time
$t-1$ for every corresponding pair of histories. For every neighbor
retained by both updates,
\[
v\in N^{\mathrm{in}}(u)\setminus[H_1],
\]
the induction hypothesis gives
\[
\hat p^{(h_1)}_{v\rightarrow H_{h_1}(C\oplus u)}(t-1)
\leq
\hat p^{(h_2)}_{v\rightarrow H_{h_2}(C\oplus u)}(t-1).
\]
Therefore,
\[
1-p_{vu}
\hat p^{(h_1)}_{v\rightarrow H_{h_1}(C\oplus u)}(t-1)
\geq
1-p_{vu}
\hat p^{(h_2)}_{v\rightarrow H_{h_2}(C\oplus u)}(t-1).
\]

Define
\[
A
=
N^{\mathrm{in}}(u)\setminus[H_1],
\]
and
\[
B
=
\left(
N^{\mathrm{in}}(u)\setminus[H_2]
\right)\setminus A.
\]
The order-$h_2$ survival product can then be written as
\[
Q_{h_2}
=
\prod_{v\in A}
\left[
1-p_{vu}
\hat p^{(h_2)}_{v\rightarrow H_{h_2}(C\oplus u)}(t-1)
\right]
\prod_{v\in B}
\left[
1-p_{vu}
\hat p^{(h_2)}_{v\rightarrow H_{h_2}(C\oplus u)}(t-1)
\right].
\]
Since every factor lies in $[0,1]$,
\[
Q_{h_2}
\leq
\prod_{v\in A}
\left[
1-p_{vu}
\hat p^{(h_2)}_{v\rightarrow H_{h_2}(C\oplus u)}(t-1)
\right].
\]
Using the message inequality above on every $v\in A$ then gives
\[
Q_{h_2}
\leq
\prod_{v\in A}
\left[
1-p_{vu}
\hat p^{(h_1)}_{v\rightarrow H_{h_1}(C\oplus u)}(t-1)
\right]
=
Q_{h_1}.
\]
Hence
\[
1-(1-p_u)Q_{h_1}
\leq
1-(1-p_u)Q_{h_2},
\]
which proves
Equation~\eqref{eq:history-order-monotonicity}. Taking the empty history
$C=()$ yields
\[
\hat p^{(h_1)}_u(t)
\leq
\hat p^{(h_2)}_u(t).
\]
Together with
Equation~\eqref{eq:full-history-upper}, this proves
Equation~\eqref{eq:dmp-order-hierarchy}.
\end{proof}

\section{Cavity Extensions}\label{app:cavity_extensions}

The cavity-message approach to the IC model described in \cite{dmp}
provides a general mechanism for suppressing immediate feedback and can
be applied to several of the approximation families considered above.
Since standard DMP can be interpreted as the cavity extension of NM, a
natural first direction is to introduce analogous messages for the
step-wise SWE approximation.

The straightforward cavity extension of SWE replaces the node-level
newly-active quantities by edge-specific quantities. The node update is
\[
\hat r_u(t)
=
\bigl(1-\hat p_u(t-1)\bigr)
\left(
1-
\prod_{v\in N(u)}
[
1-p_{vu}\hat r_{v\rightarrow u}(t-1)
]
\right),
\]
while the corresponding cavity message is
\[
\hat r_{v\rightarrow u}(t)
=
\bigl(1-\hat p_{v\rightarrow u}(t-1)\bigr)
\left(
1-
\prod_{w\in N(v)\setminus\{u\}}
\left[
1-p_{wv}\hat r_{w\rightarrow v}(t-1)
\right]
\right).
\]
We take
\[
\hat r_u(0)=p_u,
\qquad
\hat r_{v\rightarrow u}(0)=p_v,
\]
and obtain the cumulative quantities by accumulating the step-wise
estimates.

The exact step-wise formulation differs fundamentally from the
cumulative cavity formulation. In the former, the neighbor-state
distribution is conditioned on the target node $u$ having remained
inactive up to time $t-1$, whereas a cavity quantity is defined by
removing $u$ from the relevant propagation branch. These operations are
not equivalent in general: conditioning on the survival of $u$ changes
the distribution of the surrounding cascade, while deleting $u$ merely
prevents it from participating in the corresponding feedback path.

The cavity-SWE construction therefore does not arise as an exact
reformulation of the step-wise conditional probabilities. Instead, it
replaces the unavailable conditional neighbor information by
edge-specific nonbacktracking messages,
\[
\hat r_{v\rightarrow u}(t),
\]
thereby suppressing immediate feedback through $u$. We consequently
interpret cavity SWE as a heuristic feedback correction to SWE rather
than as a cavity representation of its exact conditional dynamics.

\paragraph{Cavity SWE-NO.}
The same cavity construction can be combined with the temporal noisy-OR
aggregation of SWE-NO. Let
$\hat a_{v\rightarrow u}(t)$ denote the accumulated cavity activation
quantity associated with the message $v\rightarrow u$. The newly-active
cavity quantity is computed as
\[
\hat r_{v\rightarrow u}(t)
=
\bigl(1-\hat a_{v\rightarrow u}(t-1)\bigr)
\left(
1-
\prod_{w\in N(v)\setminus\{u\}}
\left[
1-p_{wv}\hat r_{w\rightarrow v}(t-1)
\right]
\right),
\]
and is accumulated through
\[
\hat a_{v\rightarrow u}(t)
=
1-
\bigl(1-\hat a_{v\rightarrow u}(t-1)\bigr)
\bigl(1-\hat r_{v\rightarrow u}(t)\bigr).
\]
At the node level, the same construction gives
\[
\hat r_u(t)
=
\bigl(1-\hat a_u(t-1)\bigr)
\left(
1-
\prod_{v\in N(u)}
\left[
1-p_{vu}\hat r_{v\rightarrow u}(t-1)
\right]
\right),
\]
followed by
\[
\hat a_u(t)
=
1-
\bigl(1-\hat a_u(t-1)\bigr)
\bigl(1-\hat r_u(t)\bigr).
\]
Thus cavity SWE-NO combines the immediate-feedback suppression of cavity
SWE with the temporal damping introduced by SWE-NO. As in the
non-cavity version, the noisy-OR accumulation should be interpreted as
an approximation mechanism rather than as an exact probability of a
union of independent activation-time events.

\paragraph{Cavity ALE.}
ALE can be extended similarly by replacing the node-level propagation
state with edge-specific cavity layer messages. Let
$h_{u\rightarrow v}^{(t)}$ denote the propagation mass at $u$ in layer
$t$ when the immediate contribution from $v$ is excluded. Starting from
\[
h_{u\rightarrow v}^{(0)}=p_u,
\]
the cavity layer recursion is
\[
h_{u\rightarrow v}^{(t+1)}
=
\Pi_{[0,1]}
\left(
\sum_{w\in N(u)\setminus\{v\}}
p_{wu}h_{w\rightarrow u}^{(t)}
\right).
\]
The corresponding node-level layer is
\[
a_u^{(t+1)}
=
\Pi_{[0,1]}
\left(
\sum_{v\in N(u)}
p_{vu}h_{v\rightarrow u}^{(t)}
\right),
\]
and cavity ALE aggregates these layers additively,
\[
\hat p_u^{\mathrm{cALE}}(T)
=
\Pi_{[0,1]}
\left(
p_u+\sum_{t=1}^{T}a_u^{(t)}
\right).
\]
Hence the only structural change relative to ALE is that propagation
along $u\rightarrow v$ excludes the immediate reverse contribution
$v\rightarrow u$ before the next layer is formed.

\paragraph{Cavity IPL.}
The cavity extension of IPL uses exactly the same cavity propagation
layers as cavity ALE, but retains the survival-product aggregation of
IPL. Thus, with $a_u^{(t)}$ obtained from the cavity messages above,
\[
\hat p_u^{\mathrm{cIPL}}(T)
=
1-
\prod_{t=0}^{T}
\left(1-a_u^{(t)}\right).
\]
Cavity IPL therefore combines two corrections operating on different
sources of repeated influence: edge-specific cavity propagation
suppresses immediate feedback within a propagation branch, while the
IPL survival product suppresses repeated counting of contributions
arriving at different propagation depths.

\paragraph{Cavity SSS.}
Finally, the natural cavity analogue of SSS is obtained by taking the
steady-state limit of DMP. Instead of iterating the node-level SSS
equation to convergence, the edge-specific DMP messages are iterated
until a fixed point is reached,
\[
\hat p_{u\rightarrow v}^{(\infty)}
=
1-(1-p_u)
\prod_{w\in N(u)\setminus\{v\}}
\left[
1-p_{wu}\hat p_{w\rightarrow u}^{(\infty)}
\right],
\]
after which the node marginals are reconstructed as
\[
\hat p_u^{(\infty)}
=
1-(1-p_u)
\prod_{v\in N(u)}
\left[
1-p_{vu}\hat p_{v\rightarrow u}^{(\infty)}
\right].
\]
Equivalently, this is the limit of standard DMP as
$T\rightarrow\infty$, and was considered alongside finite-horizon DMP
in \cite{dmp}. Note that this method is not novel, we include its description for completion.

\paragraph{SWE-HIB-C.}
We also include a hybrid cavity variant, denoted SWE-HIB-C. This method
combines the step-wise propagation principle of SWE with the
edge-cavity representation of DMP. Let
\(\hat p_{u\to v}(t)\) denote the cumulative cavity estimate for node
\(u\) when the immediate contribution from \(v\) is excluded. We
initialize
\[
\hat p_{u\to v}(0)=p_u^{(0)},
\qquad
\hat p_{u\to v}(-1)=0.
\]
The newly accumulated cavity mass at time \(t-1\) is then
\[
\Delta \hat p_{w\to u}(t-1)
=
\hat p_{w\to u}(t-1)
-
\hat p_{w\to u}(t-2),
\]
and, for \(t=1,\ldots,T\), the cavity states are updated according to
\[
\hat p_{u\to v}(t)
=
1-
\left(1-\hat p_{u\to v}(t-1)\right)
\prod_{w\in N^{\mathrm{in}}(u)\setminus\{v\}}
\left[
1-p_{wu}\Delta\hat p_{w\to u}(t-1)
\right].
\]
Thus, as in SWE, only probability mass newly accumulated during the
preceding step is allowed to propagate, while the cavity representation
suppresses immediate feedback through the destination of the message.

After \(T\) cavity updates, the implementation performs a DMP-style
node-level readout,
\[
\hat p_u(T+1)
=
1-
(1-p_u^{(0)})
\prod_{v\in N^{\mathrm{in}}(u)}
\left[
1-p_{vu}\hat p_{v\to u}(T)
\right].
\]
Under this indexing convention, the parameter \(T\) denotes the number
of cavity-message updates, while the final node-level readout introduces
one additional propagation layer. Hence the returned estimate
corresponds to a maximum propagation depth of \(T+1\). This convention
matches the implementation used in the experiments.

\subsubsection{Relation to Higher-Order and Loop-Corrected Message Passing}
\label{app:higher-order-related-work}

The higher-order DMP construction introduced above provides a systematic
mechanism for correcting short-loop feedback in the Independent Cascade
model by retaining progressively more of the recent cavity history. The
general idea of improving message-passing approximations by enlarging the
state, retaining additional path history, or explicitly accounting for
short loops is, however, well established. Closely related constructions
have appeared in percolation theory, epidemic modeling, graphical-model
inference, and dynamic cavity methods
\cite{lin2017highorder,radicchi2016beyond,zhang2017spectral,
cantwell2019message,kirkley2021belief,mooij2007loop,
delferraro2015dynamic,gao2024recurrent}.
The contribution here is therefore not the general principle of
higher-order or loop-corrected message passing, but its formulation as a
hierarchy of finite-horizon IC marginal approximations with a bounded
rolling cavity history. In particular, the construction extends standard
IC DMP directly, admits a monotonic ordering with the retained history,
and at order two reduces to a sparse collection of nonredundant directed
triangle-closing states.

\paragraph{Higher-order non-backtracking IC bounds.}
The closest IC-specific precedent is the generalized non-backtracking
framework of Abbe et al.~\cite{nbub}. Their $r$-nonbacktracking
construction retains a recent path of vertices and excludes previously
visited vertices when bounding subsequent propagation, thereby producing
progressively tighter upper bounds as the retained path length increases.
Structurally, this is closely related to the rolling history used in our
higher-order DMP hierarchy: ordinary nonbacktracking corresponds to
retaining only the immediately preceding vertex, while additional history
suppresses feedback through progressively longer short cycles. The main
difference lies in the quantity being approximated and in the resulting
state representation. Abbe et al. construct rigorous path-based upper
bounds on the final IC influence, whereas our recursion extends the usual
DMP equations directly and produces finite-horizon node activation
marginals. Moreover, at order two we retain an additional state only when
the extra excluded vertex actually participates in the current incoming
neighborhood, so that the nonredundant correction states correspond
specifically to directed triangle closures.

\paragraph{High-order non-backtracking matrices.}
Lin et al.~\cite{lin2017highorder} introduced high-order
non-backtracking matrices for estimating percolation thresholds. Their
state space consists of directed paths of increasing length, so that
higher-order transitions prevent returns to vertices appearing further
back in the recent path. This is closely related to our use of a bounded
recent history to remove feedback paths not eliminated by standard
edge-based DMP. Their construction, however, is a linear spectral method
whose primary object is the percolation threshold, whereas ours is a
nonlinear probabilistic recursion for IC activation marginals. In
addition, our order-two implementation eliminates history states that do
not alter the cavity update rather than representing all admissible
higher-order path states explicitly.

\paragraph{Triangle-aware percolation message passing.}
Radicchi and Castellano~\cite{radicchi2016beyond} developed a
message-passing approximation for percolation that explicitly discounts
redundant propagation paths created by triangles. Their work demonstrates
the same basic failure mode motivating our second-order construction:
ordinary locally tree-like message passing treats branches connected by
short loops as if they were independent, thereby repeatedly accounting
for correlated propagation routes. Their correction is formulated
specifically around triangular motifs in percolation, while our triangle
correction arises automatically as the $h=2$ member of a general
history-based DMP hierarchy for the IC model.

\paragraph{Triangle-nonbacktracking message passing.}
Zhang~\cite{zhang2017spectral} similarly developed a triangle-aware
message-passing formulation for bond percolation and, by linearizing the
resulting equations, derived the triangle-nonbacktracking matrix. The
method incorporates triangles as explicit structural objects and improves
spectral estimation of the percolation transition on clustered networks.
Our second-order DMP addresses essentially the same type of local
dependence, but the construction is organized differently: triangles are
not introduced as distinguished motifs in advance. Instead, the
higher-order cavity recursion retains recent vertices, and a genuinely new
second-order state appears only when the retained history closes a
triangle. This also embeds the triangle correction in a hierarchy that
extends naturally to longer retained histories.

\paragraph{Message passing with loop-containing neighborhoods.}
Cantwell and Newman~\cite{cantwell2019message} proposed a systematic
family of message-passing methods that replaces the locally tree-like
assumption by calculations over increasingly large local neighborhoods
containing short loops. Increasing the neighborhood radius incorporates
progressively more local cyclic structure and can recover exact results
when the relevant loops are contained within the chosen neighborhoods.
This shares the central principle of our construction that accuracy can be
improved systematically by retaining additional local dependence
information. The retained objects are nevertheless different: their
method enlarges the local graphical region on which a message is defined,
whereas our method keeps the underlying edge-based DMP structure and
enlarges only the recent cavity history carried by the message.

\paragraph{Belief propagation on networks with loops.}
Kirkley et al.~\cite{kirkley2021belief} extended this neighborhood-based
idea to belief propagation for general probabilistic models on networks
with short loops. Their method performs inference using overlapping
loop-containing neighborhoods rather than the independent branches
assumed by ordinary belief propagation, thereby directly representing
local correlations induced by short cycles. The connection to our method
is therefore conceptual rather than model-specific: both approaches
systematically relax the locally tree-like assumption by increasing the
context represented by each message. Our construction is substantially
more specialized, using the progressive structure of IC to retain only a
bounded sequence of recently excluded vertices instead of an entire local
subgraph.

\paragraph{Loop-corrected cavity distributions.}
A more general line of work on loop-corrected belief propagation was
developed by Mooij and Kappen~\cite{mooij2007loop}. Their method first
approximates cavity distributions over the Markov blanket of each
variable and subsequently corrects these distributions by imposing
consistency constraints, explicitly targeting correlations that ordinary
belief propagation neglects in loopy factor graphs. The shared principle
is again the correction of errors caused by treating neighboring cavity
branches as independent. Their approach attempts to represent richer
joint cavity distributions and can therefore incur complexity that grows
rapidly with the local Markov blanket, whereas our IC hierarchy targets a
narrower source of dependence by carrying only finite path history.

\paragraph{Higher-order dynamic cavity approximations.}
Del Ferraro and Aurell~\cite{delferraro2015dynamic} introduced a dynamic
message-passing approximation for reversible kinetic spin systems in
which distributions over spin histories are projected onto finite-order
Markov processes. Increasing the Markov order retains progressively more
temporal history and therefore provides a hierarchy of dynamic-cavity
closures. This is related to our method in its use of bounded memory to
improve a cavity approximation, but the meaning of that memory is
different. Their higher order describes temporal dependence in recurrent
spin trajectories, whereas our history is spatial: it records recently
visited vertices of a propagation branch in order to suppress feedback
through short graph cycles.

\paragraph{Higher-order recurrent DMP with loops.}
Most directly among recent dynamic-message-passing approaches, Gao
et al.~\cite{gao2024recurrent} extend recurrent DMP for epidemic
processes by introducing higher-order cavity messages that omit multiple
vertices. Their construction explicitly addresses feedback through local
loops beyond the immediate edge-level echo removed by ordinary rDMP, and
at the next order the nontrivial additional messages are associated with
triangular structures. This is closely analogous to the mechanism of our
order-two correction. The underlying stochastic processes are
nevertheless different: their method concerns recurrent epidemic
dynamics in which nodes may repeatedly change state and requires a closure
of higher-order dynamic quantities, whereas IC is progressive and admits
a direct live-edge interpretation. Our construction consequently defines
the higher-order approximation through truncated spatial cavity history
and yields finite-horizon IC node marginals, together with the monotonic
hierarchy established above.

\section{Experimental Methodology}\label{app:exp_method}

\subsection{Real-World Networks}
\label{subsec:real-networks}

We collected several real-world networks that are commonly used in IM benchmarks.

The first dataset is NetHEPT \cite{chen2009efficient_nethept}, which is an undirected co-authorship network constructed from papers in the arXiv High Energy Physics – Theory collection, where nodes represent authors and edges represent co-authorship relationships.

WikiVote \cite{leskovec2010signed_wikivote} is a directed social network derived from Wikipedia administrator elections, where nodes represent Wikipedia users and a directed edge from one user to another indicates that the former voted in the latter's administrator election. The data cover elections from the inception of Wikipedia through January 2008.

Enron \cite{leskovec2009community_enron-slashdot} is an email communication network constructed from the Enron email corpus. Nodes correspond to email addresses, with an edge connecting two addresses if at least one email was sent between them; the SNAP representation used here is undirected.

Epinions \cite{richardson2003trust_epinions} is a directed online social network obtained from the consumer review website Epinions.com. Nodes represent users, while directed edges encode user-declared trust relationships forming the site's Web of Trust.

Slashdot \cite{leskovec2009community_enron-slashdot} is a social network derived from the technology news website Slashdot and its ``Slashdot Zoo'' feature. Nodes represent users and edges correspond to friend or foe relationships declared between users; the snapshot used here was collected in February 2009.

ca-CondMat \cite{leskovec2007graph_ca-condmat} is an undirected scientific collaboration network constructed from papers submitted to the Condensed Matter section of arXiv between January 1993 and April 2003. Nodes represent authors, and two authors are connected if they co-authored at least one paper.

cit-HepTh \cite{ForestFire-cit-HepTh} is a directed citation network consisting of papers from the High Energy Physics Theory section of arXiv between January 1993 and April 2003. Nodes represent papers, and a directed edge from one paper to another indicates that the former cites the latter.

\subsection{Synthetic Networks}
\label{subsec:synthetic-networks}

We consider five standard synthetic network models, each generated with \(n=5000\) nodes and target mean degree \(d\in\{4,10,20\}\). The generated graphs are undirected and are subsequently represented in the IC model by replacing every undirected edge with two oppositely directed arcs.

\textbf{Erdős--Rényi (ER).}
The Erdős--Rényi model \cite{erdos_renyi, gilbert1959random} generates a random graph in which every pair of nodes is connected independently with the same probability. For target mean degree \(d\), we set

$$
p=\frac{d}{n-1},
$$

so that the expected mean degree is \(d\).

\textbf{Barabási--Albert (BA).}
The Barabási--Albert model \cite{barabasi1999emergence} generates scale-free networks through preferential attachment: newly introduced nodes preferentially connect to already well-connected nodes. Each new node attaches to

$$
m=\operatorname{round}(d/2)
$$

existing nodes, yielding an asymptotic mean degree of approximately \(2m\).

\textbf{Holme--Kim (HK).}
The Holme--Kim model \cite{holme2002growing} extends preferential attachment with triad formation, producing heavy-tailed degree distributions while allowing substantially higher clustering than the standard BA model. We use

$$
m=\operatorname{round}(d/2)
$$

attachments per new node and triad-formation probability \(p_{\mathrm{triad}}=0.5\).

\textbf{Stochastic Block Model (SBM).}
The Stochastic Block Model \cite{sbm} generates networks with explicit community structure by assigning different connection probabilities within and between groups. We use four equally sized blocks and a mixing parameter \(\mu=0.3\), so that \(70\%\) of the expected degree is assigned to within-community connections and \(30\%\) to between-community connections. For block size \(n/4\),

$$
p_{\mathrm{in}}
=
\frac{(1-\mu)d}{n/4-1},
\qquad
p_{\mathrm{out}}
=
\frac{\mu d}{n-n/4}.
$$

\textbf{Watts--Strogatz (WS).}
The Watts--Strogatz model \cite{watts1998collective} generates small-world networks by starting from a regular ring lattice and randomly rewiring edges. We use initial degree \(k=d\)
and rewiring probability \(\beta=0.1\). Since the implementation requires an even lattice degree, the selected values \(d\in\{4,10,20\}\) satisfy this condition directly.

Across all five families, the benchmark therefore varies network density through \(d\in\{4,10,20\}\), while keeping \(n=5000\) fixed.


\subsection{Diffusion Configurations and Regime Calibration}
\label{subsec:diffusion-configurations}

We evaluate each network under several diffusion configurations that vary the initial seed-set size, the edge-probability model, and the overall diffusion strength. The seed set is selected uniformly at random with seed fraction

$$
s\in\{0.01,0.05,0.10,0.20\},
$$

corresponding to \(1\%\), \(5\%\), \(10\%\), and \(20\%\) of the nodes being initially active.

We consider three base edge-probability models. In the \emph{constant} model, every edge receives the same base probability,

$$
\bar p_{uv}=c,
$$

so differences in diffusion arise entirely from graph structure rather than heterogeneity in the transmission probabilities. In our implementation, the constant generator assigns the same value to every edge before the subsequent global calibration step.

The \emph{weighted-cascade} model assigns

$$
\bar p_{uv}
=
\frac{1}{d_{\mathrm{in}}(v)},
$$

where \(d_{\mathrm{in}}(v)\) is the in-degree of the destination node \(v\). Thus, edges entering high-degree nodes receive smaller probabilities. Equivalently, before global rescaling, the incoming edge probabilities of every non-isolated node sum to one. This normalizes the total incoming transmission strength across nodes of different in-degree and is commonly used in influence-maximization experiments.

Finally, the \emph{heterogeneous Dirichlet} model introduces substantial variation among the incoming transmission probabilities of a node. For each node \(v\) with \(d_{\mathrm{in}}(v)\) incoming edges, a vector of relative incoming weights is sampled from a symmetric Dirichlet distribution,

$$
(w_{1v},\ldots,w_{d_{\mathrm{in}}(v)v})
\sim
\operatorname{Dirichlet}(\kappa,\ldots,\kappa),
$$

and these weights define the corresponding base edge probabilities. We use concentration parameter \(\kappa=0.5.\)

Because \(\kappa<1\), the resulting distribution tends to be uneven: a small number of incoming edges may receive relatively large transmission probabilities while many others receive small probabilities. In contrast to the weighted-cascade model, which distributes the incoming probability mass uniformly across the predecessors of a node, the Dirichlet model preserves the same general normalization while introducing edge-level heterogeneity. The benchmark uses \(\kappa=0.5\) throughout.

Using identical edge probabilities across graphs with different topologies, densities, or seed-set sizes can result in substantially different cascade sizes. We therefore control the effective diffusion strength by calibrating a global scaling parameter separately for each graph realization, edge-probability model, and seed fraction.

Let \(\bar p_e\) denote the base probability generated by one of the three models above. For a global scaling factor \(\alpha\), the final transmission probability is

$$
p_e(\alpha)=\min\{1,\alpha \bar p_e\}.
$$

Consequently, the edge-probability models determine the \emph{relative distribution} of transmission strengths across the network, whereas \(\alpha\) controls their overall magnitude.

We characterize the resulting diffusion strength by the secondary activation rate

$$
g(\alpha)
=
\frac{1}{|V\setminus S|}
\sum_{u\notin S}
\Pr_\alpha(u\text{ is activated}\mid S),
$$

which is the expected fraction of nodes outside the initial seed set that become activated during the cascade. Measuring diffusion relative to the initially inactive population prevents larger seed sets from mechanically producing a larger reported diffusion strength.

We consider three target secondary activation rates,

$$
g^\star\in\{0.10,0.35,0.60\},
$$

representing weak, moderate, and strong diffusion regimes, respectively. Thus, \(g^\star=0.35\), for example, specifies that approximately \(35\%\) of the initially inactive nodes should subsequently become activated through diffusion.

Since IC activation probabilities are monotone non-decreasing in the edge probabilities, the scale \(\alpha\) is determined using a one-dimensional binary search. At each candidate value, \(g(\alpha)\) is estimated using a Monte Carlo pilot simulation. Once a scale satisfying the target within the prescribed tolerance is found, the resulting edge probabilities are fixed and used by all approximation methods and by the higher-budget Monte Carlo reference evaluation.

If the requested target cannot be reached even after all positive edge probabilities have saturated at one, the corresponding configuration is marked as infeasible rather than being replaced by a different diffusion regime. The calibration procedure uses \(250\) Monte Carlo realizations per candidate scale, at most \(10\) binary-search iterations, and a target tolerance of \(0.02\).

\subsection{Hyperparameter Calibration and Held-Out Evaluation}
\label{subsec:hyperparameter-protocol}

To prevent method hyperparameters from being selected on the same instances used for the final comparison, we separate every experimental setting into independent calibration and held-out test instances. For each method, the complete candidate parameter grid is evaluated on three calibration instances. The parameter configuration with the lowest mean calibration Root mean squared error (RMSE) is then selected, frozen, and evaluated without further adjustment on five held-out instances.

For the synthetic benchmark, calibration and test instances use independently generated graph realizations. Thus, the selected hyperparameters must generalize not only to new seed sets and diffusion parameters, but also to previously unseen graphs drawn from the same graph family and target mean degree. Three independent graph/diffusion instances are used for calibration and five independent instances for held-out evaluation.

For the real-world benchmark, the network topology remains fixed within each dataset. Independent instances are instead obtained from independently sampled seed sets and, for the heterogeneous Dirichlet model, independent edge-weight realizations. The constant and weighted-cascade edge models are deterministic conditional on the graph, and therefore only the seed-set realization varies between their calibration and held-out instances. The same three-calibration/five-test split is used as in the synthetic benchmark.

For methods with a propagation-horizon parameter, the candidate grid contains

$$
T\in\{1,\ldots,20\}.
$$

For SWE-NO and its cavity variant, which additionally use compensatory continuation, we consider

$$
L\in\{0,\ldots,5\},
$$

and evaluate the Cartesian product of the candidate \(T\) and \(L\) values. Methods without tunable parameters are evaluated directly.

Parameter selection is performed separately for every combination of network setting, edge-probability model, seed fraction, and target secondary activation rate. Candidate configurations are ranked by their mean RMSE across the three calibration instances. With the selection tolerance set to zero, the configuration attaining the minimum mean calibration RMSE is selected; exact ties are resolved in favor of smaller \(T\), followed by smaller continuation depth \(L\).

Monte Carlo simulation provides the reference node-level activation probabilities used both for hyperparameter calibration and for the final held-out evaluation. In the synthetic benchmark, each calibration instance is evaluated using \(10{,}000\) Monte Carlo realizations, whereas each held-out test instance uses \(50{,}000\) realizations. The larger test budget reduces sampling noise in the final reported accuracy measurements, while the smaller calibration budget limits the substantial computational cost of evaluating the complete hyperparameter grid.

For the real-world benchmark, the Monte Carlo budget is adapted to dataset size. The numbers of realizations used for calibration and held-out evaluation are

$$
\begin{array}{lcc}
\text{Dataset} & \text{Calibration} & \text{Held-out test}\\
\hline
\text{NetHEPT}   & 10{,}000 & 30{,}000\\
\text{WikiVote}  & 10{,}000 & 30{,}000\\
\text{CondMat}   &  8{,}000 & 20{,}000\\
\text{cit-HepTh} &  8{,}000 & 20{,}000\\
\text{Enron}     &  6{,}000 & 15{,}000\\
\text{Epinions}  &  5{,}000 & 10{,}000\\
\text{Slashdot}  &  5{,}000 & 10{,}000 .
\end{array}
$$

The lower budgets used for the larger real-world networks reflect the substantially greater cost of individual Monte Carlo realizations. In all cases, the reference budget used for the held-out evaluation is larger than that used during hyperparameter selection, concentrating computational effort on the measurements reported in the final comparison.

\subsection{Evaluation Metrics}
\label{subsec:evaluation-metrics}

We evaluate each approximation by comparing its vector of estimated
node activation probabilities
\[
\hat{\mathbf p}=(\hat p_u)_{u\in V}
\]
with the corresponding Monte Carlo reference
\[
\mathbf p^{\mathrm{MC}}=(p^{\mathrm{MC}}_u)_{u\in V}.
\]
We consider complementary metrics measuring absolute estimation error,
linear agreement, rank preservation, and computational cost.

Our primary absolute-error metric is the root mean squared error (RMSE),
\[
\mathrm{RMSE}
=
\sqrt{
\frac{1}{|V|}
\sum_{u\in V}
\left(\hat p_u-p^{\mathrm{MC}}_u\right)^2
}.
\]
RMSE penalizes large node-level errors more strongly than small ones and
is also used as the criterion for hyperparameter selection in the
calibration stage. For each experimental setting, the parameter
configuration minimizing mean RMSE across the calibration instances is
selected before the held-out evaluation.

We additionally report the mean absolute error (MAE),
\[
\mathrm{MAE}
=
\frac{1}{|V|}
\sum_{u\in V}
\left|\hat p_u-p^{\mathrm{MC}}_u\right|,
\]
which provides a more direct measure of the typical magnitude of the
node-level estimation error and is less sensitive than RMSE to a small
number of large deviations.

To measure whether an approximation preserves variation in activation
probabilities across nodes, we report the Pearson correlation
coefficient
\[
\rho_P
=
\operatorname{corr}
\left(
\hat{\mathbf p},
\mathbf p^{\mathrm{MC}}
\right).
\]
Pearson correlation measures linear agreement between the estimated and
reference marginals and can remain high even when an estimator exhibits
a systematic difference in scale or offset. It therefore complements,
rather than replaces, the absolute-error metrics.

We place particular emphasis on the Spearman rank correlation,
\[
\rho_S
=
\operatorname{corr}
\left(
\operatorname{rank}(\hat{\mathbf p}),
\operatorname{rank}(\mathbf p^{\mathrm{MC}})
\right).
\]
Unlike RMSE and MAE, Spearman correlation measures whether the method
preserves the relative ordering of nodes rather than their precise
probability values. This is important for rank-sensitive downstream
applications. In particular, influence-maximization procedures
ultimately rely on distinguishing and ordering competing alternatives,
so an approximation that preserves relative influence information may
remain useful even when its probability estimates exhibit some
systematic numerical bias. We emphasize, however, that the Spearman
coefficient measured here ranks node activation marginals for a fixed
diffusion instance; it is not itself a direct measurement of the
ranking of candidate marginal influence gains used by a greedy IM
algorithm.

We additionally compute the area under the receiver operating
characteristic curve (AUC) as a rank-based classification measure.
Nodes whose Monte Carlo activation probability lies at or above the
95th percentile are treated as positive instances, while the estimated
activation probabilities are used as classification scores. Thus, AUC
measures how well an approximation separates the most highly activated
nodes from the remainder of the network. Because the definition is
threshold-based and can be influenced by deterministically active seed
nodes, we treat AUC as a supplementary metric rather than a primary
measure of approximation quality.

Finally, we record the runtime of each approximation. Runtime is
reported together with accuracy because the principal motivation for
analytical IC approximations is to reduce the computational cost of
diffusion estimation relative to high-budget Monte Carlo simulation.
Consequently, no accuracy metric alone characterizes the practical
quality of an approximation; the relevant comparison is the trade-off
between predictive accuracy and computational cost.

\subsection{Implementation Details}\label{app:imp_details}
All approximation methods and Monte Carlo simulations were implemented
in Python within a common computational framework. NumPy was used for
the main numerical operations, with SciPy and scikit-learn used for
selected evaluation utilities and NetworkX used for generating the
synthetic graph topologies. The approximation methods themselves operate
on a custom compact graph representation rather than on NetworkX
objects.

Specifically, every network is converted to the same \texttt{ICGraph}
representation, in which directed edges are stored as contiguous,
source-sorted arrays of source vertices, destination vertices, and edge
probabilities. Incoming-edge indices, reverse-edge indices, and other
structural information required by particular methods are constructed
from this common representation and cached when needed. All
approximation methods therefore receive the same graph object, the same
edge-probability array, and the same initial activation vector, and
return a length-$|V|$ array of node activation-probability estimates
aligned with the same node ordering.

The implementations were designed to make runtime comparisons as
meaningful as possible. In particular, closely related methods share
the same low-level representation and, where applicable, the same
array-based aggregation operations. Incoming sums and products are
implemented using vectorized NumPy operations such as
\texttt{bincount}, while products of many probability terms are
generally evaluated in the logarithmic domain using \texttt{log},
\texttt{log1p}, and \texttt{exp}. Cavity methods use edge-indexed
message arrays over the same underlying edge representation rather than
constructing separate graph objects. Consequently, differences in
runtime primarily reflect the additional state and computation required
by the approximation itself, rather than differences between unrelated
graph libraries or implementations. Methods whose definitions require
additional operations, such as shortest-path searches or repeated
fixed-point computations, retain those operations rather than replacing
them with algorithmically different shortcuts.

Edge probabilities are stored using 32-bit floating-point values in the
benchmark graph representation. Node indices are stored using 32-bit
integers whenever the graph size permits, while structural pointer
arrays use 64-bit integer indexing. Internal arithmetic is
method-dependent: several array-based methods retain the graph's
single-precision representation, whereas methods for which additional
numerical stability is useful explicitly promote working arrays to
64-bit floating point. Probability products are clipped where required
to avoid invalid logarithms and floating-point overflow or underflow.
Thus, all methods receive numerically identical edge-probability inputs,
although their internal working precision is not artificially forced to
be identical.

Runtime is measured as wall-clock execution time using
\texttt{time.perf\_counter}. For the held-out benchmark, timing covers
only the execution of the approximation with its previously selected
hyperparameters; graph loading, diffusion-regime calibration,
hyperparameter selection, Monte Carlo reference generation, and metric
computation are excluded. This separation is important because the
reported runtime is intended to measure the cost of applying an
approximation to an already specified IC instance rather than the cost
of constructing the experimental instance itself. All methods are
evaluated on the same machine and within the same benchmark process.

For second-order DMP, the directed triangle-state indexing depends only
on graph topology and can be reused across diffusion instances.
Accordingly, this topology-dependent indexing is constructed and cached
before the timed approximation call, so the reported runtime measures
the iterative propagation stage rather than the one-time structural
preprocessing cost. We therefore interpret the reported runtimes as
\emph{per-evaluation runtimes on a prepared graph}; preprocessing costs
that can be amortized across multiple diffusion evaluations are not
included in this measure.

All experiments were run on a Dell Precision 7710 workstation with an
Intel Core i7-6820HQ processor (8 logical CPUs), 16 GB of RAM, and a
256.1 GB disk, running Ubuntu 24.04.4 LTS.

The real-world datasets used in the experiments are publicly available
from the sources cited in Section~\ref{subsec:real-networks}, while all synthetic
networks can be regenerated from the graph models, parameters, and
random seeds specified above. The complete implementation, and benchmark
scripts is available at our repository \cite{IC_CODE}.




\end{appendices}




\bibliography{sn-bibliography}

\end{document}